\documentclass[journal,onecolumn]{IEEEtran}
\usepackage[T1]{fontenc}
\usepackage{amsmath,amsfonts,amsthm, mathabx}
\usepackage{algorithmic,lipsum}
\usepackage{array}
\usepackage[caption=false,font=normalsize,labelfont=sf,textfont=sf]{subfig}
\usepackage{textcomp,tcolorbox}
\usepackage{stfloats}
\usepackage{hyperref}       % hyperlink
\hypersetup{
  colorlinks   = true, %Colors links instead of ugly boxes
  urlcolor     = blue, %Color for external hyperlinks
  linkcolor    = blue, %Color of internal links
  citecolor   = blue %Color of citations, could be ``red''
}
\usepackage{url}
\usepackage{float}
\usepackage{verbatim}
\usepackage{graphicx}
\def\BibTeX{{\rm B\kern-.05em{\sc i\kern-.025em b}\kern-.08em
    T\kern-.1667em\lower.7ex\hbox{E}\kern-.125emX}}
\usepackage{bbm}
\usepackage{physics}
\newtheorem{theorem}{Theorem}[section]
\newtheorem{corollary}[theorem]{Corollary}
\newtheorem{lemma}[theorem]{Lemma}
\newtheorem{remark}[theorem]{Remark}
\newtheorem{definition}[theorem]{Definition}
\newtheorem{prop}[theorem]{Proposition}

\newtheorem{conjecture}[theorem]{Conjecture}
\usepackage{tikz-cd}
\usepackage{tikz}
\usetikzlibrary{decorations.pathreplacing}
\usepackage[
  left=0.5in,
  right=0.5in,
  top=0.5in,
  bottom=0.5in,
]{geometry}
\usetikzlibrary{arrows, automata}
\usetikzlibrary{shapes.geometric, arrows}
\tikzstyle{block} = [rectangle, rounded corners, minimum width=2.5cm, minimum height=1cm, text centered, draw=black, fill=blue!20]
\tikzstyle{arrow} = [thick,->,>=stealth]
\tikzstyle{shield} = [rectangle, rounded corners, minimum width=2.5cm, minimum height=1cm, text centered, draw=black, fill=red!20]

\usepackage{color,hyperref}
\usepackage{physics}

\newcommand{\wt}{\widetilde}

\newcommand{\mb}{\mathbb}
\newcommand{\mc}{\mathcal}
\newcommand{\om}{\omega}
\newcommand{\ot}{\otimes}

\newcommand{\Choi}{Choi–Jamio\l{}kowski\ }

\newcommand{\sym}{\mathrm{sym}}
\newcommand{\asym}{\mathrm{asym}}

\definecolor{cool_green}{rgb}{0.0, 0.5, 0.0}

\newcommand{\pw}[1]{{\color{red}[PW: #1]}}

\newcommand{\added}[1]{\textcolor{blue}{#1}\normalcolor}

\title{Quantum capacity amplification via privacy}
\begin{document}
\author{%
Peixue Wu$^{1,2,*}$ and Yunkai Wang$^{1,2,3}$ \\[1ex]
\small{
\begin{tabular}{l}
$^{1}$\textit{Institute for Quantum Computing, University of Waterloo,}
\textit{200 University Avenue West, Waterloo, ON N2L 3G1, Canada} \\[1ex]
$^{2}$\textit{Department of Applied Mathematics, University of Waterloo,}
\textit{200 University Avenue West, Waterloo, ON N2L 3G1, Canada} \\[1ex]
$^{3}$\textit{Perimeter Institute for Theoretical Physics}
\textit{31 Caroline St N, Waterloo, ON N2L 2Y5, Canada} \\[1ex]
\end{tabular}
}
\thanks{*Contact author: p33wu@uwaterloo.ca.}
}
\maketitle

\begin{abstract}
We investigate superadditivity of quantum capacity through private channels whose Choi–Jamiołkowski operators are private states. This perspective links the security structure of private states to quantum capacity and clarifies the role of the shield system: information encoded in the shield system that would otherwise leak to the environment can be recycled when paired with an assisting channel, thereby boosting capacity. Our main contributions are threefold: Firstly, we develop a general framework that provides a sufficient condition for capacity amplification, which is formulated in terms of the assisting channel’s Holevo information. As examples, we give explicit, dimension and parameter dependent amplification thresholds for erasure and depolarizing channels. Secondly, assuming the Spin alignment conjecture, we derive a single-letter expression for the quantum capacity of a family of private channels that are neither degradable, anti-degradable, nor PPT; as an application, we construct channels with vanishing quantum capacity yet unbounded private capacity. Thirdly, we further analyze approximate private channels: we give an alternative proof of superactivation that extends its validity to a broader parameter regime, and, by combining amplification bounds with continuity estimates, we establish a metric separation showing that channels exhibiting capacity amplification have nonzero diamond distance from the set of anti-degradable channels, indicating that existing approximate-(anti)degradability bounds are not tight. We also revisit the computability of the regularized quantum capacity and modestly suggest that this fundamental question still remains open.
\end{abstract}

\tableofcontents

\newpage
\section{Introduction}
\subsection{Background and Motivation}
Quantum channels, the fundamental objects describing information transmission in quantum mechanics, have attracted intensive study in a variety of contexts such as quantum computing, quantum cryptography, and quantum communication. One of the most intriguing effects for quantum channels is the super-additivity of quantum capacities given by regularized coherent information \cite{lloyd1997capacity, shor2002, devetak2005capacity}. This effect enables innovative and counterintuitive phenomena, with instances demonstrated for depolarizing channels~\cite{lloyd1997capacity, bhalerao2025}; constructions based on perturbative methods~\cite{leditzky2018dephrasure, leditzky2023generic, leditzky2023platypus, siddhu2020leaking, siddhu2021positivity, Smith_2025, wuzhen2025_2, wuzhen2025_1}; and superadditivity for Gaussian channels~\cite{lim2018activation, lim2019activation}. Nevertheless, a comprehensive theory is still lacking, and fundamental questions about the scope and structure of superadditivity remain open.

In this work, we study the superadditivity of quantum capacities in the context of channels induced by private states, which provide a general framework for investigating these phenomena. Introduced in \cite{Horodecki_2005, Horodecki_2009}, private states represent a natural quantum extension of secret classical correlations, augmented by a ``shield'' system. Formally, a private state $\gamma^{a_0 b_0 A_0 B_0}$ is composed of two principal components: the $a_0 b_0$ subsystem, designated for classical key distribution between two parties (Alice and Bob), and the $A_0 B_0$ subsystem, which serves to obscure these correlations from an eavesdropper. A notable special case is the \emph{pbit} (private bit), corresponding to a two-dimensional key system $a_0, b_0 \cong \mathbb{C}^2$. Perfect private states exemplify this framework by exhibiting classical correlations on $a_0 b_0$ that are entirely decoupled from any eavesdropper, serving as an archetype for secure quantum communication.

Viewing a private state as the Choi--Jamio\l kowski operator of a suitable quantum channel $\mathcal{N}^{A\to B}$ (which we call a \emph{private channel}) with $A=a_0A_0$ and $B=b_0B_0$, allows one to translate structural features of private state $\gamma^{a_0 b_0 A_0 B_0}$ directly into statements about its quantum capacity. This \emph{channel perspective} is powerful: it ties the security guarantees of private states to core questions in capacity theory. Intuitively, information stored in the shield leaks to the environment and is therefore useless for quantum communication. However, when combined with a second channel, the leaked information can be \emph{recycled} to enhance end-to-end transmission. This intuition underpins the phenomena of \emph{super-activation}, \emph{amplification}, and \emph{super-amplification} in the channel setting.

\textit{Super-activation}, first discovered by Smith and Yard~\cite{Smith_2008}, occurs when two channels, each individually having zero quantum capacity, can be combined (via the tensor product) to produce a strictly positive capacity. If only one of the channels has positive capacity and combining them yields a strictly higher capacity than that single channel alone, we refer to the effect as \emph{amplification}. If both channels have nonzero capacities but their combination exceeds the sum of those individual capacities, we call it \emph{super-amplification}. For a comprehensive discussion of these effects, which underscore the profoundly non-additive nature of quantum capacity, contrasting with the classical capacity of classical channels, see Ref.~\cite{Koudia_2022}.  

Beyond demonstrating that capacities can be nonadditive, \emph{amplification} is practically important: it shows how to enhance the ultimate rate of reliable quantum communication by pairing channels. Prior work has typically isolated only one slice of this landscape: Ref.~\cite{siddhu2021entropic} verified superadditivity of the \emph{maximal coherent information} (a one-shot quantity), and Ref.~\cite{Smith_2008} established \emph{superactivation} (two zero-capacity channels combining to yield positive capacity). In contrast, we rigorously present a unified framework that exhibits \emph{super-activation}, \emph{amplification}, and \emph{super-amplification} within the same family of constructions, together with dimension- and parameter-dependent thresholds. This idea is closely related to the \emph{potential capacity} of a quantum channel, introduced in Ref.~\cite{Winter_2016}, asking for the maximum capacity a channel can attain when used together with an arbitrary assisting channel. We complement this general notion with concrete examples: in our constructions, we identify channels whose \emph{potential quantum capacity} strictly exceeds their standalone quantum capacity, and we quantify the gap in terms of the shield/key parameters that appear naturally in the private state picture.

%--------------------- Main results- --------------------------------------

\subsection{Our Contributions}
We summarize our contributions in three parts.

\paragraph{A general framework for quantum capacity amplification}
A guiding question is: for a private (or approximately private) channel $\mc N$ and an arbitrary channel $\mc M$, under what conditions do we have
\begin{equation}\label{intro: amplification}
    \mc Q(\mc N \otimes \mc M)\;>\;\mc Q(\mc N)\;+\;\mc Q(\mc M)\,?
\end{equation}
Here $\mc Q(\cdot)$ is the quantum capacity.
\begin{theorem}[Informal; see Theorem~\ref{thm:amplification}]
A sufficient condition relating the Holevo information of $\mc M$ to $\mc Q(\mc N)+\mc Q(\mc M)$ implies \eqref{intro: amplification}.
\end{theorem}
We illustrate this condition for erasure and depolarizing channels, leveraging recent upper bounds on quantum capacity (e.g., \cite{Fanizza_2020, Zhu_2025}) to obtain explicit, dimension- and parameter-dependent amplification thresholds.

\paragraph{Gap between private capacity and quantum capacity}
In this section, we address a fundamental question: how large can the gap between a channel’s private and quantum capacities be? Since coherent transmission and privacy are closely related, it was initially conjectured that the two capacities would coincide. However, Horodecki et al.~\cite{Horodecki_2005} and Smith-Yard~\cite{Smith_2008} showed that certain channels are too noisy to transmit quantum information yet can still support private communication, although they did not quantify how large this gap could be. Leung et al.~\cite{Leung_2014} later demonstrated that the gap can grow with the input dimension, while the quantum capacity remains bounded below by a constant, see also~\cite{Smith_2009}. Building on this line of work, we construct an even stronger example where the quantum capacity vanishes while the private capacity diverges, providing a sharper manifestation of their fundamental separation.
\begin{theorem}[Section~\ref{subsec:construction}]
There exist a sequence of channels $\{\mc M_n=\mc M_n^{A_n\to B_n}\}_{n\ge 1}$ such that
\[
    \mc Q(\mc M_n)=\frac{1}{n}\to 0,\qquad \mc P(\mc M_n)=n\to \infty,
\]
where $\mc P(\cdot)$ is the private capacity.
\end{theorem}
The construction of this example relies on deriving a single-letter formula for the quantum capacity of private channels. Owing to their structure, information encoded in the shield leaks to the environment, leaving only the private subsystem to carry quantum information. By formalizing this intuition through the \emph{Spin Alignment Conjecture} (SAC)~\cite{leditzky2023platypus}, we obtain a single-letter expression for the quantum capacity of a family of private channels that are neither degradable, anti-degradable, nor PPT.

% \paragraph{Single-letter quantum capacity for private channels}
% Because of the structure of private channels, information encoded into the shield leaks to the environment and does not contribute to coherent transmission; only the private subsystem can carry quantum information. We formalize this intuition via the \emph{Spin Alignment Conjecture} (SAC) \cite{leditzky2023platypus} and obtain a single-letter formula for the quantum capacity of a family of private channels that are neither degradable, anti-degradable, nor PPT, which are the main categories of channels with single-letter quantum capacity. As an application, we construct channels with anomalous behavior:
% \begin{theorem}[Section~\ref{subsec:construction}]
% There exist a sequence of channels $\{\mc M_n=\mc M_n^{A_n\to B_n}\}_{n\ge 1}$ such that
% \[
%     \mc Q(\mc M_n)=\frac{1}{n}\to 0,\qquad \mc P(\mc M_n)=n\to \infty,
% \]
% where $\mc P(\cdot)$ is the private capacity.
% \end{theorem}
% This gives a distinct realization of the extensive gap between private and quantum capacities compared to Ref.~\cite{Leung_2014}, with the key feature that the quantum capacity vanishes in the limit, while the quantum capacity remains lower bounded by a constant in Ref.~\cite{Leung_2014}.

\paragraph{Approximate private channels and their applications}
Finally, we develop a robust, private-state–based mechanism for capacity amplification. First, we prove quantitative \emph{amplification} bounds for $\varepsilon$–approximate private channels when paired with a quantum erasure channel (Proposition~\ref{prop:amplification 2}), turning approximate privacy into positive coherent information in a controlled regime and yielding explicit assisted lower bounds on quantum capacity. This in particular provides an alternative proof of the superactivation effect \cite{Smith_2008}. Second, combining these bounds with continuity estimates, we obtain a \emph{metric} separation: channels exhibiting  capacity amplification effect are at nonzero diamond distance from the set of anti-degradable channels, clarifying the approximate-(anti)degradability bound for quantum capacities provided in Ref.~\cite{Sutter_2017} is loose for this class of channels.  Third, we revisit the construction in Ref.~\cite{Cubitt_2015} concerning the fundamental question of computability of the regularized quantum capacity. In addition to some further simplification of their construction, we also notice an important fact leading to the remark below:

\medskip
\noindent\textit{Computability remark.}
The constructions in Ref.~\cite{Sutter_2017}, while insightful, do \emph{not} resolve whether the regularized quantum capacity
$\mc Q(\mc N)=\lim_{k\to\infty}\frac{1}{k}\,\mc Q^{(1)}\!\big(\mc N^{\otimes k}\big)$
is computable. By “computable” we mean that for any channel $\mc N$, there exists a finite $N\ge 1$ such that
$\mc Q(\mc N)=\frac{1}{N}\,\mc Q^{(1)}\!\big(\mc N^{\otimes N}\big)$.
Neither our construction nor Ref.~\cite{Cubitt_2015} rules out this possibility. In fact, Ref.~\cite{Cubitt_2015} shows that for each fixed $N\ge1$ there exists a (dependently constructed) channel $\mc N_N$ with strict superadditivity at $N$:
\[
\mc Q(\mc N_N)>\frac{1}{N}\,\mc Q^{(1)}\!\big(\mc N_N^{\otimes N}\big).
\]
However, for such $\mc N_N$ it remains possible that some larger $K>N$, we have 
\[
\mc Q(\mc N_N)=\frac{1}{K}\,\mc Q^{(1)}\!\big(\mc N_N^{\otimes K}\big).
\]

%-------------------Organization-----------------------------
\subsection{Organization of This Paper}
\begin{itemize}
\item \textbf{Section~\ref{sec:prelim}.}
We review channel notation and representations (Kraus, Choi), coherent information and its basic properties (data processing, direct-sum, and flagged-channel rules), and the structure of private states and perfect pbits, including the coherent-information lemmas used throughout.

\item \textbf{Section~\ref{sec:amp}.}
We develop a general framework for \emph{quantum capacity amplification}. Building on private-state induced channels, we prove quantitative amplification bounds (e.g., Proposition~\ref{prop:amplification 2}), give explicit \emph{amplification} and \emph{super-amplification} examples, and illustrate the conditions for erasure and depolarizing channels.

\item \textbf{Section~\ref{sec:single letter}.}
Assuming the Spin Alignment Conjecture (SAC) \cite{leditzky2023platypus}, we show a single-letter formula for the quantum capacity of private channels in a specific regime, even though the channels are neither degradable, anti-degradable, nor PPT. As an application, we construct a family of channels exhibiting a vanishing quantum capacity alongside an unbounded private capacity.

\item \textbf{Section~\ref{sec:approximate}.}
We analyze approximate private channels. First, we obtain a metric separation from the anti-degradable set (diamond-norm lower bounds) via the amplification effect explored in our work and continuity estimates. Second, we give an alternative superactivation proof in the approximate setting (extending the Smith–Yard phenomenon \cite{Smith_2008}). Third, we revisit the construction in Ref.~\cite{Cubitt_2015} concerning the computability of quantum capacity. We clarify that these constructions do \emph{not} resolve computability of the regularized quantum capacity.
\end{itemize}
%-------------------------- Prelim-----------------------------------------
\section{Preliminary}\label{sec:prelim}
\textbf{Notation.}
\begin{itemize}
    \item Capital letters $A,B,C,E$ (``Alice, Bob, Charlie, Eve'') denote finite-dimensional Hilbert spaces with dimensions $d_A,d_B,\ldots$.
    \item $\mc B(A,B)$ is the space of linear operators from $A$ to $B$; we write $\mc B(A):=\mc B(A,A)$. The identity operator on $A$ is $\mb I^A$; the identity superoperator on $\mc B(A)$ is $id^{A\to A}$ (often just $id$).
    \item States (density operators) on $AB$ are denoted $\rho^{AB}$. Superscripts indicate the subsystems on which an operator acts nontrivially.
    \item Linear maps (superoperators) $\mc N^{A\to B}:\mc B(A)\to\mc B(B)$ denote quantum channels if they are completely positive and trace-preserving (CPTP).
    \item The trace $\tr$ is taken over the indicated subsystem, e.g.\ $\tr_E$; the partial transpose on $B$ is $(\cdot)^{T_B}$ with respect to a fixed computational basis.
    \item We use $\|\cdot\|_1$ for trace norm and $\|\cdot\|_\diamond$ for the diamond norm.
\end{itemize}

%----------------------------- Basics on channels ----------------------------
\subsection{Quantum channels and their representations}

\paragraph{Stinespring dilation and complementary channels.}
Let $A,B,E$ be finite-dimensional Hilbert spaces. An isometry $V:A\to B\otimes E$ (so $V^\dagger V=\mb I^A$) induces a pair of CPTP maps
\begin{equation}\label{eq:stinespring}
    \mc N^{A\to B}(\rho) = \Tr_E  \bigl[V\rho V^\dagger\bigr],
    \qquad
    (\mc N^c)^{A\to E}(\rho) = \Tr_B  \bigl[V\rho V^\dagger\bigr],
\end{equation}
called \emph{complementary channels}. Every channel admits such a dilation, unique up to local unitaries.

\paragraph{Kraus representation.}
A channel $\mc N^{A\to B}$ admits an operator-sum (Kraus) decomposition
\begin{equation}\label{eq:Kraus}
    \mc N(X) = \sum_{i=1}^m K_i X K_i^\dagger,\qquad K_i\in \mc B(A,B),\qquad \sum_{i=1}^m K_i^\dagger K_i=\mb I^A.
\end{equation}
The minimal number $m$ of Kraus operators equals the rank of the Choi operator (below).

\paragraph{Choi--Jamio\l{}kowski representation.}
Fix an orthonormal basis $\{\ket{i}\}_{i=0}^{d_A-1}$ of $A$ and let
\begin{equation}\label{eq:maxent}
    \ket{\Psi}^{AA'} := \frac{1}{\sqrt{d_A}}\sum_{i=0}^{d_A-1}\ket{i}_A\otimes\ket{i}_{A'},
\end{equation}
where $A'\cong A$. The (normalized) \Choi operator of $\mc N^{A'\to B}$ is
\begin{equation}\label{eq:CJ}
    J_{\mc N}^{AB} := (id_A\otimes \mc N^{A'\to B})  \bigl(\ket{\Psi}\bra{\Psi}^{AA'}\bigr) \in \mc B(A\otimes B).
\end{equation}
Equivalently, in the chosen basis,
\begin{equation}\label{eq:CJ-matrix}
    J_{\mc N} = \frac{1}{d_A}\sum_{i,j=0}^{d_A-1} \ketbra{i}{j}_A \otimes \mc N(\ketbra{i}{j}_{A'}).
\end{equation}
Well-known equivalences:
\begin{equation}\label{eq:CP-TP-CJ}
    \mc N \text{ CP } \Longleftrightarrow J_{\mc N}\ge 0,
    \qquad
    \mc N \text{ TP } \Longleftrightarrow \Tr_B(J_{\mc N})=\frac{\mb I^A}{d_A}.
\end{equation}
(Some authors use the unnormalized \Choi operator $\mc J_{\mc N}:=d_A J_{\mc N}$, which then satisfies $\Tr_B(\mc J_{\mc N})=\mb I^A$.)

\paragraph{Reconstruction and reshuffling identities.}
The action of $\mc N$ can be recovered from its \Choi operator via
\begin{equation}\label{eq:reconstruct}
    \mc N(X) =  d_A\,\Tr_A  \bigl[\,J_{\mc N}\,(X^{T}\otimes \mb I^B)\,\bigr],
\end{equation}
where ${}^{T}$ denotes matrix transpose in the basis of \eqref{eq:maxent}.  
Conversely, if $J\in\mc B(A\otimes B)$ satisfies $J\ge 0$ and $\Tr_B(J)=\mb I^A/d_A$, then \eqref{eq:reconstruct} defines a unique CPTP map with \Choi operator $J$.

\paragraph{From \Choi to Kraus (one convenient choice).}
Let the spectral decomposition be $J_{\mc N}=\sum_{k=1}^r \lambda_k \ketbra{\psi_k}{\psi_k}^{AB}$ with $\lambda_k>0$ and $\ket{\psi_k}\in A\otimes B$. Write
$\ket{\psi_k}=\sum_{i,j} c^k_{ij}\,\ket{i}_A\otimes \ket{j}_B$ and set
\begin{equation}\label{eq:CJ-to-Kraus}
    K_k = \sqrt{d_A\,\lambda_k} \sum_{i,j} c^k_{ij}\,\ket{j}_B\bra{i}_A  \in  \mc B(A,B).
\end{equation}
Then $\mc N(X)=\sum_{k=1}^r K_k X K_k^\dagger$ and $\sum_k K_k^\dagger K_k=\mb I^A$. (If you use the unnormalized \Choi operator $\mc J_{\mc N}=d_A J_{\mc N}$, drop the factor $d_A$ in \eqref{eq:CJ-to-Kraus}.)

\paragraph{Norms and distances.}
For later use we recall the diamond norm
\begin{equation}\label{eq:diamond}
    \|\mc N-\mc M\|_\diamond := \sup_{X\neq 0} 
    \frac{\|(\mc N\otimes   id_R)(X)-(\mc M\otimes   id_R)(X)\|_1}{\|X\|_1},
\end{equation}
where $R$ is any system with $\dim R\ge d_A$; the supremum is attained on a purification. Contractivity under CPTP post-processing implies
\begin{equation}\label{eq:diamond-contract}
    \|\mc N\otimes \mc E-\mc M\otimes \mc E\|_\diamond  \le  \|\mc N-\mc M\|_\diamond
    \qquad \text{for every CPTP }\mc E.
\end{equation}
The \Choi–diamond relation is $\|\mc N-\mc M\|_\diamond  \le  d_A\,\|J_{\mc N}-J_{\mc M}\|_1$.

%-------------------------- Capacities-----------------------------------------
\subsection{Capacities of quantum channels and their properties}
Suppose a complementary pair $(\mc N,\mc N^c)$ is generated by an isometry
$V_{\mc N}:A\to B\otimes E$ as in \eqref{eq:stinespring}. The \emph{quantum capacity} $\mc Q(\mc N)$
is the supremum of all achievable rates for reliable quantum information transmission through $\mc N$.
By the Lloyd–Shor–Devetak (LSD) theorem~\cite{lloyd1997capacity, shor2002, devetak2005private},
the coherent information is an achievable rate.

For any input state $\rho^A\in \mc B(A)$, let $\ket{\psi}^{RA}$ be a purification on $R\otimes A$,
and define $\ket{\Psi}^{RBE}:=(\mb I^R\otimes V_{\mc N})\ket{\psi}^{RA}$.
Let $\rho^{RB}:=\Tr_E\ketbra{\Psi}{\Psi}^{RBE}$ and $\rho^{E}:=\Tr_{RB}\ketbra{\Psi}{\Psi}^{RBE}$.
The coherent information is
\begin{equation}\label{def:coherent information}
\begin{aligned}
    I_c(\rho^A,\mc N)
    &:= I(R\rangle B)_{\rho^{RB}}
      =  S(\rho^B)-S(\rho^{RB})
      =  S(\rho^B)-S(\rho^{E}),
\end{aligned}
\end{equation}
where $S(\cdot)$ is the von Neumann entropy. Different choices of purification yield the same value.
For brevity we often write $S(\rho^B)$ as $S(B)$. The one-shot coherent information is
\begin{equation}\label{def:max coherent information}
    \mc Q^{(1)}(\mc N) = \max_{\rho^A} I_c(\rho^A,\mc N),
\end{equation}
and the LSD theorem gives the regularized capacity
\[
\mc Q(\mc N) = \lim_{n\to\infty}\frac{1}{n}\,\mc Q^{(1)}(\mc N^{\otimes n}).
\]

The \emph{private information} of $\mc N$ for an ensemble $\{p_x,\rho_x^A\}$ is
\[
    I_p(\{p_x,\rho_x^A\},\mc N) :=  I(\mc X;B) - I(\mc X;E),
\qquad
    \mc P^{(1)}(\mc N) := \sup_{\{p_x,\rho_x^A\}} I_p(\{p_x,\rho_x^A\},\mc N).
\]
The private (classical) capacity $\mc P(\mc N)$ is given by \cite{devetak2005private}
\[
\mc P(\mc N) = \lim_{n\to\infty}\frac{1}{n}\,\mc P^{(1)}(\mc N^{\otimes n}).
\]

\paragraph{Direct sum of channels and flagged channels}
Let $\{\Phi_k^{A_k\to B_k}\}_{k=1}^n$ be channels, the \emph{direct sum} channel
\[
  \bigoplus_{k=1}^n \Phi_k^{A_k\to B_k}: 
  \mc B \Bigl(\bigoplus_{k=1}^n A_k\Bigr)\longrightarrow \mc B \Bigl(\bigoplus_{k=1}^n B_k\Bigr)
\]
acts block-diagonally (off-diagonal blocks are sent to $0$; see~\cite{Fukuda_2007}, and~\cite{Chessa_2021,wuzhen2025_3} for generalizations).
Explicitly, for $X=\sum_{k,l} \ket{k}\bra{l} \otimes X_{k\ell}$,
\[
  \Bigl(\,\bigoplus_{k=1}^n \Phi_k^{A_k\to B_k}\Bigr)(X)
   = 
  \begin{pmatrix}
    \Phi_1(X_{11}) & 0 & \cdots & 0 \\
    0 & \Phi_2(X_{22}) & \cdots & 0 \\
    \vdots & \vdots & \ddots & \vdots \\
    0 & 0 & \cdots & \Phi_n(X_{nn})
  \end{pmatrix}.
\]
Given a classical flag register $F$ with basis $\{\ket{i}^F\}_{i=0}^{d_F-1}$,
a channel $\mc N^{A\to FB}$ is \emph{flagged} if
\[
    \mc N^{A\to FB} = \sum_{i=0}^{d_F-1} p_i\,\ketbra{i}{i}^F\otimes \mc N_i^{A\to B},
    \qquad p_i\ge 0, \sum_i p_i=1.
\]
A canonical example is the (binary) \emph{erasure channel} with parameter $\lambda\in[0,1]$:
\begin{equation}\label{eqn:erasure channel}
    \mc E_{\lambda}^{A\to FA}
     = 
    (1-\lambda)\,\ketbra{0}{0}^F\otimes id^{A\to A}
     + 
    \lambda\,\ketbra{1}{1}^F\otimes \mc E_{1}^{A\to A},
\end{equation}
where $\mc E_{1}$ maps every input to a fixed state on $A$.

\paragraph{Coherent information of direct-sum and flagged channels}
Recall the definition \eqref{def:coherent information} and the one-shot quantity \eqref{def:max coherent information}. We also use the shorthand
\begin{equation}\label{def:Q^1}
    \mc Q^{(1)}(\mc N^{A\to B})  =  \max_{\rho^A} I_c(\rho^A,\mc N^{A\to B}),
\end{equation}
which is consistent with \eqref{def:max coherent information}.
The following properties are standard.

\begin{lemma}\label{lemma:basic}
    Let $\mc N=\mc N_0\oplus \mc N_1$ and $n\ge 1$. Then
    \begin{equation}\label{coh:switch}
        \mc Q^{(1)}(\mc N^{\otimes n})
         = 
        \max_{0\le \ell\le n} 
        \mc Q^{(1)} \bigl(\mc N_0^{\otimes \ell}\otimes \mc N_1^{\otimes (n-\ell)}\bigr).
    \end{equation}
    If $\mc N^{A\to FB}=\sum_{i=0}^{d_F-1} p_i\,\ketbra{i}{i}^F\otimes \mc N_i^{A\to B}$ is a flagged channel, then for every input $\rho^A$,
    \begin{equation}\label{coh:flag}
        I_c\bigl(\mc N^{A\to FB},\rho^A\bigr)
         = 
        \sum_{i=0}^{d_F-1} p_i\, I_c\bigl(\mc N_i^{A\to B},\rho^A\bigr).
    \end{equation}
\end{lemma}

\begin{proof}
For the direct sum, expand $(\mc N_0\oplus \mc N_1)^{\otimes n}$ over bit-strings $\mathbf b\in\{0,1\}^n$ to get
$\bigoplus_{\mathbf b}\bigotimes_{j=1}^n \mc N_{b_j}$. Since coherent information of a direct sum equals
the maximum over the summands~\cite[Prop.~1]{Fukuda_2007}, and the order of tensor factors is irrelevant,
only the Hamming weight $\ell$ of $\mathbf b$ matters, giving \eqref{coh:switch}.
For the flagged channel, the output is a classical-quantum mixture with orthogonal flags, then using the entropy formula for probabilistic mixture of orthogonal states: 
\begin{align*}
      S(\sum_{i} p_i \tau_i) = H(\{p_i\}) + \sum_{i} p_i S(\tau_i),
\end{align*}
where $\{p_i\}$ is a probability distribution and $\{\tau_i\}$ is a set of orthogonal states, the coherent information is additive under flag mixture of channels. 
\end{proof}

\paragraph{Data processing for coherent information.}
For any bipartite state $\rho^{AB}$ and channel $\mc P^{B\to C}$,
\begin{equation}\label{DPI:coherent}
    I(A\rangle B)_{\rho^{AB}}
     \ge 
    I(A\rangle C)_{(id_A\otimes \mc P)(\rho^{AB})},
\end{equation}
see~\cite[Thm.~11.9.3]{wilde2011classical}. As immediate corollaries,
for channels $\mc N_1^{A\to B}$ and $\mc N_2^{B\to C}$ and any input $\rho^A$,
\begin{equation}\label{DPI:channel}
    I_c(\mc N_2 \circ \mc N_1,\rho^A)
     \le 
    \min\Bigl\{\,I_c(\mc N_2,\mc N_1(\rho^A)),\ I_c(\mc N_1,\rho^A)\Bigr\},
\end{equation}
and for a tripartite state $\rho^{AB_1B_2}$,
\begin{equation}\label{DPI:remove}
    I(A\rangle B_1B_2)_{\rho^{AB_1B_2}}
     \ge 
    I(A\rangle B_1)_{\rho^{AB_1}}.
\end{equation}
An equality case in \eqref{DPI:remove} is useful. Let $\mc E_1^{B_2\to B_2}$ be a \emph{replacement} (complete erasure) channel,
\[
\mc E_1(X)\;=\;\Tr(X)\,\sigma^{B_2},
\]
for some fixed state $\sigma^{B_2}$. Then for any state $\rho^{A B_1 B_2}$,
\begin{equation}\label{DPI:equality}
    I(A\rangle B_1 B_2)_{(id^{A B_1}\otimes \mc E_1^{B_2\to B_2})(\rho^{A B_1 B_2})}
    \;=\;
    I(A\rangle B_1)_{\rho^{A B_1}}.
\end{equation}

%-------------------------------------- Private states -----------------------------
\subsection{Private states}
We consider a four-party mixed state $\gamma^{a_0b_0A_0B_0}$ with $\mathrm{dim}\,a_0=\mathrm{dim}\,b_0=d_0$ and $\mathrm{dim}\,A_0=\mathrm{dim}\,B_0=d$.
\begin{itemize}
    \item $a_0,A_0$ belong to Alice and $b_0,B_0$ belong to Bob. We denote
    \[
        A:=a_0A_0,\quad B:=b_0B_0.
    \]
    \item The subsystem $a_0b_0$ is called the \textit{key system}.
    \item The subsystem $A_0B_0$ is called the \textit{shield system}.
\end{itemize}
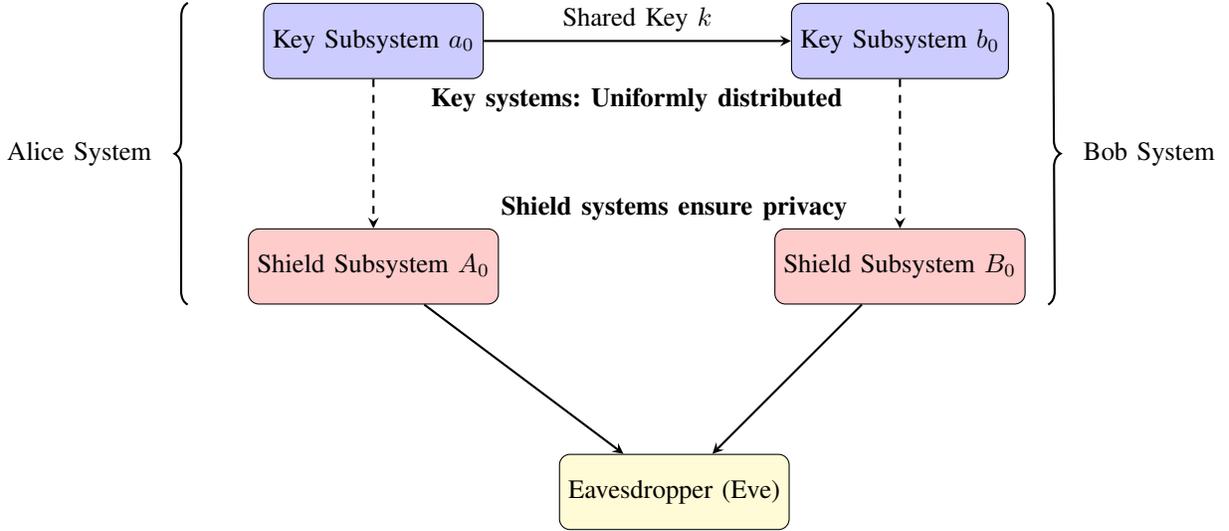
\begin{figure}[ht]
\begin{center}
\begin{tikzpicture}[node distance=3cm]

% Alice's system
\node (alice_key) [block] {Key Subsystem $a_0$};
\node (alice_shield) [shield, below of=alice_key] {Shield Subsystem $A_0$};

% Bob's system
\node (bob_key) [block, right of=alice_key, xshift=4cm] {Key Subsystem $b_0$};
\node (bob_shield) [shield, below of=bob_key] {Shield Subsystem $B_0$};

% Eavesdropper's system
\node (eve) [block, below of=alice_shield, xshift=4cm, fill=yellow!20] {Eavesdropper (Eve)};

% Connections
\draw [arrow] (alice_key) -- (bob_key) node[midway, above] {Shared Key $k$};
\draw [arrow] (alice_shield) -- (eve) node[midway, right] {};
\draw [arrow] (bob_shield) -- (eve) node[midway, right] {};

% Key-shield connections
\draw [arrow, dashed] (alice_key) -- (alice_shield) node[midway, left] {};
\draw [arrow, dashed] (bob_key) -- (bob_shield) node[midway, right] {};

% Labels for security
\node at (4,-1.95) [below] {\textbf{Shield systems ensure privacy}};
\node at (3.5,-0.5) [below] {\textbf{Key systems: Uniformly distributed}};

% Curly bracket for Alice's system
\draw [decorate, decoration={brace, amplitude=5pt, mirror}, thick] 
  ([xshift=-1cm]alice_key.north west) -- ([xshift=-0.8cm]alice_shield.south west) 
  node[midway, left=10pt] {Alice System};

% Curly bracket for Bob's system
\draw [decorate, decoration={brace, amplitude=5pt}, thick] 
  ([xshift=0.5cm]bob_key.north east) -- ([xshift=0.3cm]bob_shield.south east) 
  node[midway, right=10pt] {Bob System};
\end{tikzpicture}
\end{center}
\caption{Illustration of a perfect private state showing the key and shield subsystems for Alice and Bob. The shared key $a_0b_0$ is protected by shield subsystems \( A_0 \) and \( B_0 \), ensuring privacy against an eavesdropper (Eve).}
\label{fig:private-state}
\end{figure}

To formalize security, let $\ket{\gamma}^{ABE}=\ket{\gamma}^{a_0A_0b_0B_0E}$ be a purification of $\gamma^{a_0b_0A_0B_0}$. The system $E$ is the eavesdropper. Measure $\ket{\gamma}^{ABE}$ in the computational basis $\{\ket{ij}^{a_0b_0}:0\le i,j\le d_0-1\}$ on $a_0b_0$, followed by tracing out the shield $A_0B_0$. The resulting ccq state is
\begin{align}
 \widetilde\gamma^{a_0b_0E}
 &= \mathrm{tr}_{A_0B_0}\Bigl(\sum_{i,j=0}^{d_0-1} \ketbra{ij}{ij}^{a_0b_0} \otimes \bra{ij}^{a_0b_0} \ketbra{\gamma}{\gamma}^{a_0A_0b_0B_0E} \ket{ij}^{a_0b_0}\Bigr) \\
 &= \sum_{i,j=0}^{d_0-1} p_{ij}\, \ketbra{ij}{ij}^{a_0b_0} \otimes \rho_{ij}^E. \label{eqn:state after measurement}
\end{align}

\begin{definition}
We say that $\gamma^{a_0b_0A_0B_0}$ with key $a_0b_0$ and shield $A_0B_0$ is \emph{secure} if the state in \eqref{eqn:state after measurement} factorizes with $E$:
\[
\widetilde\gamma^{a_0b_0E}=\sum_{i,j=0}^{d_0-1} p_{ij}\, \ketbra{ij}{ij}^{a_0b_0} \otimes \rho^E.
\]
Moreover, $\gamma^{a_0b_0A_0B_0}$ is a \emph{perfect private state} if the ccq state has the form
\[
\widetilde\gamma^{a_0b_0E}=\sum_{i=0}^{d_0-1} \frac{1}{d_0}\, \ketbra{ii}{ii}^{a_0b_0} \otimes \rho^E.
\]
\end{definition}
\noindent To emphasize the key dimension, a perfect private state is called a \textit{pdit}; when $d_0=2$ it is a \textit{pbit}.

Now consider the quantum channel associated with a pbit. Let the maximally entangled state on $AA':=a_0A_0a_0'A_0'$ be
\[
\ket{\Psi}^{AA'}=\ket{\Psi}^{a_0A_0a_0'A_0'}=\frac{1}{\sqrt{d_0 d}}\sum_{i=0}^{d_0-1}\sum_{j=0}^{d-1} \ket{ij}^{a_0A_0}\otimes \ket{ij}^{a_0'A_0'}.
\]
There exists a channel $\mc N^{A'\to B}$ such that
\[
\gamma^{a_0b_0A_0B_0}=(id^{A\to A}\otimes \mc N^{A'\to B})(\Psi^{a_0A_0a_0'A_0'}).
\]
Realize $\mc N^{A'\to B}$ as an isometry $U_{\mc N}:A'\to BE$ we have $\ket{\gamma}^{ABE}=(id^{A\to A}\otimes U_{\mc N})(\ket{\Psi}^{AA'})$. 

We recall several features of perfect private bits~\cite{Horodecki_2005, Horodecki_2009}.
\begin{prop}
$\gamma^{a_0b_0A_0B_0}$ is a perfect private state with key $a_0b_0$ and shield $A_0B_0$ if it is of the form \cite{Horodecki_2005}
\begin{equation}\label{pbit:general form}
\gamma^{a_0b_0A_0B_0}=\frac{1}{d_0}\sum_{k,l=0}^{d_0-1} \ketbra{k}{l}^{a_0}\otimes \ketbra{k}{l}^{b_0}\otimes U_k^{A_0B_0}\,\sigma^{A_0B_0}\,(U_l^{A_0B_0})^\dagger
\end{equation}
for some mixed state $\sigma^{A_0B_0}$ and unitaries $U_k^{A_0B_0}$, $0\le k\le d_0-1$.
\end{prop}

\begin{lemma}\label{pbit:basic}
Suppose $\gamma^{a_0b_0A_0B_0}$ is a perfect pbit. Then
\begin{align}
    & I(a_0\rangle b_0A_0B_0)_{\gamma^{a_0b_0A_0B_0}}=1, \label{pbit:equality 1}\\
    & I(a_0\rangle b_0)_{\gamma^{a_0b_0}}=1-h\Bigl(\frac{1+|c|}{2}\Bigr)\ge 0, \label{pbit:equality 2}
\end{align}
where
\begin{align}\label{pbit:quantity 1}
    h(x):=-x\log_2 x-(1-x)\log_2(1-x),\quad c:=\mathrm{tr}\bigl(U_0^{A_0B_0}\,\sigma^{A_0B_0}\,(U_1^{A_0B_0})^\dagger\bigr).
\end{align}
\end{lemma}

\begin{proof}
Let $\ket{\psi}^{a_0b_0}=\frac{1}{\sqrt2}(\ket{00}+\ket{11})^{a_0b_0}$ and define
$U^{a_0b_0A_0B_0}:=\sum_{k=0}^1 \ketbra{kk}{kk}^{a_0b_0}\otimes U_k^{A_0B_0}$.
Then \eqref{pbit:general form} can be written as
\[
\gamma^{a_0b_0A_0B_0}=U^{a_0b_0A_0B_0}\bigl(\ketbra{\psi}{\psi}^{a_0b_0}\otimes \sigma^{A_0B_0}\bigr)(U^{a_0b_0A_0B_0})^\dagger.
\]
Moreover,
\[
\gamma^{b_0A_0B_0}=\frac{1}{2}\sum_{k=0}^1 \ketbra{k}{k}^{b_0}\otimes U_k^{A_0B_0}\,\sigma^{A_0B_0}\,(U_k^{A_0B_0})^\dagger.
\]
Unitary invariance of entropy gives
\[
I(a_0\rangle b_0A_0B_0)_{\gamma}=S(\gamma^{b_0A_0B_0})-S(\gamma^{a_0b_0A_0B_0})=S(\sigma^{A_0B_0})+1-S(\ketbra{\psi}{\psi}^{a_0b_0}\otimes \sigma^{A_0B_0})=1,
\]
proving \eqref{pbit:equality 1}. For \eqref{pbit:equality 2}, note that
\[
\gamma^{a_0b_0}=\frac{1}{2}\bigl(\ketbra{00}{00}^{a_0b_0}+\ketbra{11}{11}^{a_0b_0}+c\,\ketbra{00}{11}^{a_0b_0}+\overline{c}\,\ketbra{11}{00}^{a_0b_0}\bigr).
\]
Hence $S(\gamma^{b_0})=h(\tfrac12)=1$ and $S(\gamma^{a_0b_0})=h(\frac{1+|c|}{2})$, yielding
\[
I(a_0\rangle b_0)=S(\gamma^{b_0})-S(\gamma^{a_0b_0})=1-h\Bigl(\frac{1+|c|}{2}\Bigr)\ge 0.
\]
\end{proof}

A typical example of private states is as follows. Let $\ket{\psi_+}=\frac{1}{\sqrt2}(\ket{00}+\ket{11})$ and $\ket{\psi_-}=\frac{1}{\sqrt2}(\ket{00}-\ket{11})$.
Let $F^{A_0B_0}=\sum_{i,j=0}^{d-1}\ketbra{ij}{ji}^{A_0B_0}$ be the swap operator on $A_0B_0$, and let $\mb I^{A_0B_0}$ be the identity on $A_0B_0$.
Define the projectors onto the symmetric and antisymmetric subspaces of $\mb C^d\otimes \mb C^d$ by
\[
P_{\sym}^{A_0B_0}=\frac{1}{2}(\mb I^{A_0B_0}+F^{A_0B_0}),
\qquad
P_{\asym}^{A_0B_0}=\frac{1}{2}(\mb I^{A_0B_0}-F^{A_0B_0}).
\]
Then the main example we will study in this paper is
\begin{equation}\label{eqn:Werner}
    \gamma_{q,d}^{a_0b_0A_0B_0} = q \ketbra{\psi_+}{\psi_+}^{a_0b_0} \otimes \frac{1}{d_{\text{sym}}} P_{\mathrm{sym}}^{A_0B_0} + (1-q) \ketbra{\psi_-}{\psi_-}^{a_0b_0} \otimes \frac{1}{d_{\mathrm{asym}}} P_{\mathrm{asym}}^{A_0B_0},\ q\in [0,1],
\end{equation}
where \begin{align*}
    d_{\text{sym}}=\frac{d(d+1)}{2},\quad d_{\text{asym}}=\frac{d(d-1)}{2}.
\end{align*}
To see why it is a pbit, we rewrite the above state in the form~\eqref{pbit:general form}. First, there exists a unitary $U^{A_0B_0}$ diagonalizing $F^{A_0B_0}$ into $\mathrm{diag}(\mb I_{d(d+1)/2},-\mb I_{d(d-1)/2})$, so that
\[
P_{\sym}^{A_0B_0}=U^{A_0B_0}\begin{pmatrix}\mb I_{d(d+1)/2}&0\\[2pt]0&0\end{pmatrix}(U^{A_0B_0})^\dagger,\quad
P_{\asym}^{A_0B_0}=U^{A_0B_0}\begin{pmatrix}0&0\\[2pt]0&\mb I_{d(d-1)/2}\end{pmatrix}(U^{A_0B_0})^\dagger.
\]
Then the standard form in \eqref{pbit:general form} is obtained by choosing
\[
\sigma^{A_0B_0}=
\begin{pmatrix}
\frac{2q}{d(d+1)}\,\mb I_{d(d+1)/2} & 0 \\
0 & \frac{2(1-q)}{d(d-1)}\,\mb I_{d(d-1)/2}
\end{pmatrix},
\quad
U_0^{A_0B_0}:=U^{A_0B_0},
\quad
U_1^{A_0B_0}:=U^{A_0B_0}
\begin{pmatrix}
\mb I_{d(d+1)/2}&0\\[2pt]0&-\mb I_{d(d-1)/2}
\end{pmatrix}.
\]

%%%%%%%%%%%%%%%%%%%%%%%%%%%%%%%%%%%%%%%%%%%%%%%%%%%%%%%%%%%%%%%%
% General non-additivity

\section{Quantum capacity amplification for private channels}\label{sec:amp}
In this section, we provide a general criteria for quantum capacity amplification for private channels. We focus on the private channel $\mc N_{q,d}$ induced by the private state \eqref{eqn:Werner} with $a_0 = b_0 = \mb C^2$, and $A_0 = B_0 = \mb C^d$. To be more specific, suppose $\ket{\psi}^{a_0a_0'}$ and $\ket{\Psi}^{A_0A_0'}$ are maximally entangled states on the bipartite system $a_0a_0'$ and $A_0A_0'$ respectively. The quantum channel $\mc N_{q,d} = \mc N_{q,d}^{A' \to B}$ with $A' = a_0'A_0'$ and $B = b_0B_0$ is determined by
\begin{equation}\label{eqn:channel}
    (id_{a_0A_0} \otimes \mc N_{q,d}^{A' \to B}) (\ketbra{\psi}{\psi}^{a_0a_0'} \otimes \ketbra{\Psi}{\Psi}^{A_0A_0'}) = \gamma_{q,d}^{a_0b_0A_0B_0}. 
\end{equation}
In the matrix form, we have 
\begin{align}\label{eqn:channel matrix general}
    \mc N_{q,d} \begin{pmatrix}
    X_{00} & X_{01} \\
    X_{10} & X_{11}
\end{pmatrix} =  \begin{pmatrix}
    q \frac{X_{00}^T + \tr(X_{00}) I_d}{d+1} + (1-q)\frac{ - X_{00}^T + \tr(X_{00}) I_d}{d-1} & q \frac{X_{01}^T + \tr(X_{01}) I_d}{d+1} - (1-q)\frac{ - X_{01}^T + \tr(X_{01}) I_d}{d-1}\\
    q \frac{X_{10}^T + \tr(X_{10}) I_d}{d+1} - (1-q)\frac{ - X_{10}^T + \tr(X_{10}) I_d}{d-1} & q \frac{X_{11}^T + \tr(X_{11}) I_d}{d+1} + (1-q)\frac{ - X_{11}^T + \tr(X_{11}) I_d}{d-1}
\end{pmatrix},\ X_{ij} \in \mc B(\mb C^d).
\end{align}

A central question in this section is:
\begin{tcolorbox}
    Given a quantum channel $\mc M = \mc M^{A_0 \to C_0}$, under which condition, we have
\begin{equation}\label{amplification}
    \mc Q(\mc N_{q,d} \otimes \mc M) > \mc Q(\mc M) + \mc Q(\mc N_{q,d}). 
\end{equation}
\end{tcolorbox}
Our main result is a sufficient condition on the channel $\mc M$ and the parameters $q,d$ such that \eqref{amplification} holds. 
\begin{theorem}\label{thm:amplification}
    Suppose $\mc M^{A_0 \to C_0}$ is a quantum channel, denote the quantum states
\begin{align}\label{eqn:sym and asym Choi}
    \mc J^{\mathrm{sym}}_{\mc M} = \frac{1}{d_{\text{sym}}} (\mc M\otimes id_{B_0}) (P_{\mathrm{sym}}^{A_0B_0}),\quad \mc J^{\mathrm{asym}}_{\mc M} = \frac{1}{d_{\text{asym}}} (\mc M\otimes id_{B_0}) (P_{\mathrm{asym}}^{A_0B_0}).
\end{align}
Then $ \mc Q(\mc N_{q,d} \otimes \mc M) > \mc Q(\mc M) + \mc Q(\mc N_{q,d})$ holds provided 
\begin{align}\label{sufficient:general}
    1-h(q) + H(\mc J_\mc M^{\sym}, \mc J_\mc M^\asym;q,1-q) > \mc Q(\mc N_{q,d}) + \mc Q(\mc M),
\end{align}
where $h(q)$ is the binary entropy, and $$H(\mc J_\mc M^{\sym}, \mc J_\mc M^\asym;q,1-q) = S(q \mc J^{\mathrm{sym}}_{\mc M} + (1-q)\mc J^{\mathrm{asym}}_{\mc M}) - qS(\mc J^{\mathrm{sym}}_{\mc M}) - (1-q) S(\mc J^{\mathrm{asym}}_{\mc M})$$
is the Holevo information.
\end{theorem}
In the remaining sections, we provide the proof of the main theorem and illustrate this theorem using different channels $\mc M$.

\begin{figure}[ht]
    \centering
    \includegraphics[width=0.5\linewidth]{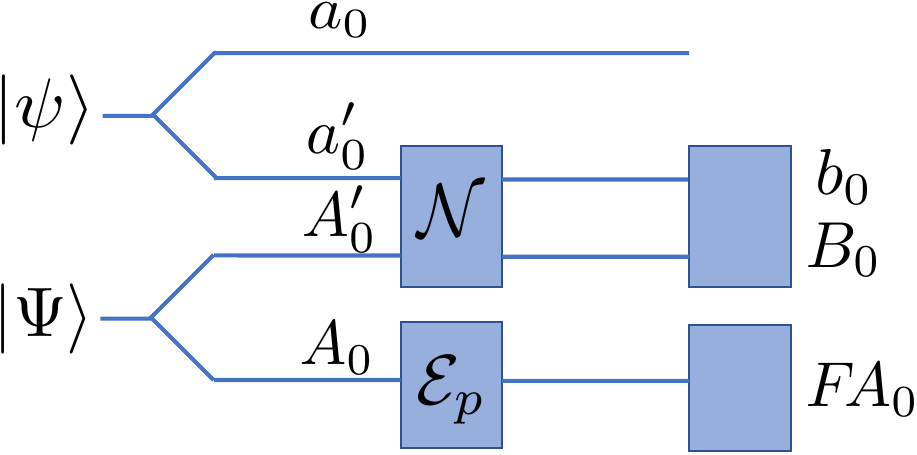}
    \caption{Quantum capacity amplification with the help of shield system.}
    \label{fig:amplification-one copy}
\end{figure}

\subsection{Lower bound on the quantum capacity of the joint channel}
To establish the result, we first derive a lower bound using maximally entangled state $\ket{\psi}^{a_0a_0'} \otimes \ket{\Psi}^{A_0A_0'}$ as an ansatz state for $\mc Q^{(1)}(\mc N_{q,d} \otimes \mc M) $. This gives us the following criteria: 
\begin{prop}\label{prop:lower bound}
    Suppose $\mc M^{A_0 \to C_0}$ is a quantum channel, denote the quantum states
\begin{align}
    \mc J^{\mathrm{sym}}_{\mc M} = \frac{1}{d_{\mathrm{sym}}} (\mc M\otimes id_{B_0}) (P_{\mathrm{sym}}^{A_0B_0}),\quad \mc J^{\mathrm{asym}}_{\mc M} = \frac{1}{d_{\mathrm{asym}}} (\mc M\otimes id_{B_0}) (P_{\mathrm{asym}}^{A_0B_0}).
\end{align}
Then we have 
\begin{equation}
\begin{aligned}
    \mc Q(\mc N_{q,d} \otimes \mc M) & \ge \mc Q^{(1)}(\mc N_{q,d} \otimes \mc M) \\
    &\ge 1-h(q) + S(q \mc J^{\mathrm{sym}}_{\mc M} + (1-q)\mc J^{\mathrm{asym}}_{\mc M}) - qS(\mc J^{\mathrm{sym}}_{\mc M}) - (1-q) S(\mc J^{\mathrm{asym}}_{\mc M}).
\end{aligned}
\end{equation}
\end{prop}
\begin{proof}
    Note that $\mc Q^{(1)}(\mc N_{q,d} \otimes \mc M)  \ge I_c(\mc N_{q,d} \otimes \mc M, \rho^{A'A_0})$, where the ansatz state $\rho^{A'A_0}$ is given by 
    \begin{align*}
        \rho^{A'A_0} = \tr_{a_0} \left(\ketbra{\psi}{\psi}^{a_0a_0'} \otimes \ketbra{\Psi}{\Psi}^{A_0A_0'}\right). 
    \end{align*}
    To calculate $I_c(\mc N_{q,d} \otimes \mc M, \rho^{A'A_0})$, we denote 
\begin{align*}
    \rho^{a_0b_0C_0B_0} & = (id_{a_0b_0B_0} \otimes \mc M) (\gamma_{q,d}^{a_0b_0A_0B_0})\\
    & = q \ketbra{\psi_+}{\psi_+}^{a_0b_0} \otimes \frac{1}{d_{\text{sym}}} (\mc M\otimes id_{B_0}) (P_{\mathrm{sym}}^{A_0B_0}) + (1-q) \ketbra{\psi_-}{\psi_-}^{a_0b_0} \otimes \frac{1}{d_{\mathrm{asym}}} (\mc M\otimes id_{B_0})(P_{\mathrm{asym}}^{A_0B_0}),
\end{align*}
thus we have 
\begin{align*}
    I_c(\mc N_{q,d} \otimes \mc M, \rho^{A'A_0}) = I(a_0 \rangle b_0C_0B_0)_{\rho^{a_0b_0C_0B_0}} = S(\rho^{b_0C_0B_0}) - S(\rho^{a_0b_0C_0B_0}). 
\end{align*}
Taking the partial trace, we have
\begin{align*}
    \rho^{b_0C_0B_0} = \frac{\mb I_2}{2} \otimes \left(q \mc J^{\mathrm{sym}}_{\mc M} + (1-q)\mc J^{\mathrm{asym}}_{\mc M}\right).
\end{align*}
The coherent information is then calculated as 
\begin{align*}
    I_c(\mc N_{q,d} \otimes \mc M, \rho^{A'A_0}) &= S(\rho^{b_0C_0B_0}) - S(\rho^{a_0b_0C_0B_0}) \\
    & = 1-h(q) + S(q \mc J^{\mathrm{sym}}_{\mc M} + (1-q)\mc J^{\mathrm{asym}}_{\mc M}) - qS(\mc J^{\mathrm{sym}}_{\mc M}) - (1-q) S(\mc J^{\mathrm{asym}}_{\mc M}).
\end{align*}
\end{proof}
Therefore, via Proposition \ref{prop:lower bound} we conclude the proof of Theorem~\ref{thm:amplification}. The following subsection establishes the upper bound on the quantum capacity of $\mc N_{q,d}$. 

\begin{comment}
    \added{ 
Unclear question:
Suppose $\mc M$ is a $d$ dimensional quantum channel, how does the Holevo information
\begin{equation}
    S\left((1-q)(id\otimes \mc M)(\frac{P_{\mathrm{sym}}}{d_{\text{sym}}}) + q (id\otimes \mc M)(\frac{P_{\mathrm{asym}}}{d_{\text{asym}}})\right) -(1-q) S\left((id\otimes \mc M)(\frac{P_{\mathrm{sym}}}{d_{\text{sym}}}) \right) + q S\left((id\otimes \mc M)(\frac{P_{\mathrm{asym}}}{d_{\text{asym}}})\right) 
\end{equation}
relate to the quantum capacity of $\mc M$? Maybe it is related to the transposition bound?
}
\end{comment}
%%%%%%%%%%%%%%%%%%%%%%%%%%%% Transposition bound
\subsection{Upper bound on the quantum capacity of each individual channel}
There are many works on the upper bound of quantum capacity, an incomplete list includes \cite{Fanizza_2020,Zhu_2024,Zhu_2025}. A classic approach is the well-known transposition bound, which states $\mc Q(\mc N^{A' \to B}) \le \log \|T_B \circ \mc N^{A' \to B}\|_{\diamond}$, see \cite{holevo2001evaluating}. The quantity $\|T_B \circ \mc N^{A' \to B}\|_{\diamond}$ can be computed using the follow SDP \cite[Theorem 3.1]{watrous2012}:
\begin{equation}\label{eqn:transposition}
    \begin{aligned}
        \|T_B \circ \mc N^{A' \to B}\|_{\diamond} & = \min \frac{1}{2}(\|Y^A\|_{op} + \|Z^{A}\|_{op}) \\
        & s.t.\ Y^{AB}, Z^{AB} \ge 0, \\
        & \hspace{0.5cm} \begin{pmatrix}
            Y^{AB} &  -\widehat{\mc J}_{T_B \circ \mc N^{A' \to B}} \\
            -\widehat{\mc J}_{T_B \circ \mc N^{A' \to B}} & Z^{AB}
        \end{pmatrix} \ge 0.
    \end{aligned}
\end{equation}
Here, $\|\cdot\|_{op}$ denotes the operator norm(largest singular value), and for any superoperator $\mc N$, $\widehat{\mc J}_{\mc N}:= \sum_{i,j}\ketbra{i}{j} \otimes \mc N(\ketbra{i}{j})$ denotes the unnormalized Choi operator. 
The upper bound on quantum capacities for general private channels with flagged forms is thus given as follows:  
\begin{prop}\label{prop:transposition bound general}
    Suppose the channel $\mc N^{a_0'A_0' \to b_0B_0}$ has unnormalized Choi–Jamio\l{}kowski operator $\mc J_\mc N$ given by
\begin{align*}
     \mc J_{\mc N}/2d = q \ketbra{\psi_+}{\psi_+}^{a_0b_0} \otimes \sigma_1^{A_0B_0} + (1-q) \ketbra{\psi_-}{\psi_-}^{a_0b_0} \otimes \sigma^{A_0B_0}_2,\quad \sigma^{A_0B_0}_1\perp \sigma^{A_0B_0}_2.
\end{align*}
%\quad \tr_{B_0}(q \sigma^{A_0B_0}_1 + (1-q)\sigma^{A_0B_0}_2) = \frac{\mb I_{A_0}}{d}.
Then an upper bound on the quantum capacity for $\mc N^{a_0A_0 \to b_0B_0}$ is 
    \begin{align}\label{upper bound on capacity: Werner}
   \mc Q(\mc N) \le \log\left( d \left\|\tr_{B_0} \left( \left|(q \sigma^{A_0B_0}_1 + (1-q)\sigma^{A_0B_0}_2)^{T_{B_0}} \right| +  \left|(q \sigma^{A_0B_0}_1 - (1-q)\sigma^{A_0B_0}_2)^{T_{B_0}} \right|\right) \right\|_{op} \right).
\end{align}
\end{prop}
\begin{proof}
    To compute \eqref{eqn:transposition}, we need to construct $Y^{AB},Z^{AB}\ge 0$ such that 
\begin{align*}
    \begin{pmatrix}
        Y^{AB} & -\mc J^{T_B} \\
        -\mc J^{T_B} & Z^{AB}
    \end{pmatrix} \ge 0.
\end{align*}
Rewriting the unnormalized Choi–Jamio\l{}kowski operator $\mc J_\mc N$ and taking the partial transpose, one has 
\begin{align*}
    \mc J_{\mc N}^{T_B}/2d &= \frac{1}{2} (\ketbra{00}{00} + \ketbra{11}{11}) \otimes (q \sigma_1 + (1-q)\sigma_2)^{T_{B_0}} + \frac{1}{2} (\ketbra{01}{10} + \ketbra{10}{01}) \otimes (q \sigma_1 - (1-q)\sigma_2)^{T_{B_0}}.
\end{align*}
A simple way to construct $Y,Z\ge 0$ is to choose  
\begin{align*}
    Y^{AB} = Z^{AB} & = d (\ketbra{00}{00} + \ketbra{11}{11}) \otimes |(q \sigma_1 + (1-q)\sigma_2)^{T_{B_0}}| + d (\ketbra{01}{01} + \ketbra{10}{10}) \otimes |(q \sigma_1 - (1-q)\sigma_2)^{T_{B_0}}|.
\end{align*}
It is straightforward to verify that \begin{align*}
    \begin{pmatrix}
        Y^{AB} & -\mc J^{T_B} \\
        -\mc J^{T_B} & Z^{AB}
    \end{pmatrix} \ge 0.
\end{align*}
$Y^A$ is calculated as
\begin{align*}
    Y^A = \tr_{b_0B_0}(Y^{AB}) = d\, I_2 \otimes \tr_{B_0} \left( |(q \sigma_1 + (1-q)\sigma_2)^{T_{B_0}}| +  |(q \sigma_1 - (1-q)\sigma_2)^{T_{B_0}}|\right).
\end{align*}
Then following the expression \eqref{eqn:transposition}, the upper bound on the quantum capacity is given by 
\begin{align*}
    \log (\|Y^A\|_{op}) = \log\left( d \left\|\tr_{B_0} \left( \left|(q \sigma^{A_0B_0}_1 + (1-q)\sigma^{A_0B_0}_2)^{T_{B_0}} \right| +  \left|(q \sigma^{A_0B_0}_1 - (1-q)\sigma^{A_0B_0}_2)^{T_{B_0}} \right|\right) \right\|_{op} \right).
\end{align*}
\end{proof}
As an application, we get an explicit upper bound for the quantum channel $\mc N_{q,d}$ induced by the private state \eqref{eqn:Werner}: 
\begin{corollary}\label{cor:upper bound private channel}
    An upper bound of the quantum capacity of $\mc N_{q,d}$ is given by 
    \begin{align*}
    \mc Q(\mc N_{q,d}) \le \log \left((d^2-1)(r_0 + |r_1|) +|r_0 + d r_1| + |r_1 + d r_0|\right),
\end{align*}
with \begin{align*}
    r_0 = \frac{q}{d(d+1)} + \frac{1-q}{d(d-1)},\quad r_1 = \frac{q}{d(d+1)} - \frac{1-q}{d(d-1)}.
\end{align*}
\end{corollary}
\begin{proof}
    We apply Proposition \ref{prop:transposition bound general}. In this case, the states are given by \eqref{eqn:Werner}: 
    \begin{align*}
    \sigma^{A_0B_0}_1 = \frac{1}{d(d+1)} \sum_{i,j} \ketbra{ij}{ij} + \ketbra{ij}{ji},\quad \sigma^{A_0B_0}_2 = \frac{1}{d(d-1)} \sum_{i,j} \ketbra{ij}{ij} - \ketbra{ij}{ji}.
\end{align*}
Therefore, we have
\begin{align*}
    (q \sigma_1 + (1-q)\sigma_2)^{T_{B_0}} = r_0 \mb I_{d^2} + d\,r_1 \Psi^+, \quad (q \sigma_1 - (1-q)\sigma_2)^{T_{B_0}} = r_1 \mb I_{d^2} + d\,r_0 \Psi^+, 
\end{align*}
where $\Psi^+ = \ketbra{\Psi^+}{\Psi^+},\ \ket{\Psi^+} = \frac{1}{\sqrt d} \sum_{i=0}^{d-1} \ket{ii}$ and
\begin{align*}
    r_0 = \frac{q}{d(d+1)} + \frac{1-q}{d(d-1)},\quad r_1 = \frac{q}{d(d+1)} - \frac{1-q}{d(d-1)}.
\end{align*}
Taking the absolute value of the matrices, 
\begin{align*}
    \left|(q \sigma_1 + (1-q)\sigma_2)^{T_{B_0}}\right| = r_0 \mb I_{d^2} + (|r_0 + d\,r_1| - r_0) \Psi^+, \quad \left|(q \sigma_1 - (1-q)\sigma_2)^{T_{B_0}}\right| = |r_1| \mb I_{d^2} + (|r_1 + d\,r_0| - |r_1|) \Psi^+.
\end{align*}
Finally, the reduced operator is given by 
\begin{align*}
    \tr_{B_0} \left( |(q \sigma_1 + (1-q)\sigma_2)^{T_{B_0}}| +  |(q \sigma_1 - (1-q)\sigma_2)^{T_{B_0}}|\right) = \left((r_0 + |r_1|)(d-\frac{1}{d}) + \frac{|r_0 + d r_1| + |r_1 + d r_0|}{d} \right) \mb I_d.
\end{align*}
Plugging it into \eqref{upper bound on capacity: Werner}, we conclude the proof.
\end{proof}

\subsection{Examples}
In this subsection, we illustrate the framework Theorem~\ref{thm:amplification}. To proceed, given a channel $\mc M^{A_0\to C_0} $, one needs an upper bound on $\mc Q(\mc M)$ and to compute the entropies of $\mc J_{\mc M}^{\sym}$ and $\mc J_{\mc M}^{\asym}$, which are defined in \eqref{eqn:sym and asym Choi}.
\subsubsection*{Erasure channels}
When $\mc M = \mc E_{\lambda,d}$, where $\mc E_{\lambda,d}$ is an erasure channel to $d+1$ dimensional output with flag $\ket{e}$: 
\begin{equation}\label{erasure channel}
    \mc E_{\lambda,d}(\rho) = (1-\lambda) \rho + \tr(\rho) \ketbra{e}{e}.
\end{equation}
Then one can directly compute $\mc J_{\mc E_{\lambda,d}}^{\sym}$ and $\mc J_{\mc E_{\lambda,d}}^{\asym}$:
\begin{align*}
    & \mc J_{\mc E_{\lambda,d}}^{\sym} = (1-\lambda) \frac{P_{\sym}}{d_{\sym}} + \lambda \ketbra{e}{e} \otimes \frac{\mb I_d}{d}, \\
    & \mc J_{\mc E_{\lambda,d}}^{\asym} = (1-\lambda) \frac{P_{\asym}}{d_{\asym}} + \lambda \ketbra{e}{e} \otimes \frac{\mb I_d}{d}.
\end{align*}
Using the entropy formula for probabilistic mixture of orthogonal states: 
\begin{equation}\label{eqn:entropy orthogonal}
    S(\sum_{i} p_i \tau_i) = H(\{p_i\}) + \sum_{i} p_i S(\tau_i),
\end{equation}
where $\{p_i\}$ is a probability distribution and $\{\tau_i\}$ is a set of orthogonal states, the entropy difference in~\eqref{sufficient:general} is
\begin{align*}
    & S(q \mc J^{\mathrm{sym}}_{\mc E_{\lambda,d}} + (1-q)\mc J^{\mathrm{asym}}_{\mc E_{\lambda,d}}) - qS(\mc J^{\mathrm{sym}}_{\mc E_{\lambda,d}}) - (1-q) S(\mc J^{\mathrm{asym}}_{\mc E_{\lambda,d}}) \\
    & = \left[(1-\lambda)q \log d_{\sym} + (1-\lambda)(1-q) \log d_{\asym} + \lambda \log d + h(\lambda) + (1-\lambda) h(q)  \right] \\
    & \hspace{0.5cm} - q \left[ h(\lambda) + (1-\lambda) \log d_\sym +\lambda \log d  \right] - (1-q) \left[ h(\lambda) + (1-\lambda) \log d_\asym +\lambda \log d  \right] \\
    & = (1-\lambda) h(q).
\end{align*}

Therefore, using Theorem \ref{thm:amplification}, the well-known fact that $\mc Q(\mc E_{\lambda,d}) = \max\{(1-2\lambda) \log d , 0\}$ and the upper bound for $\mc Q(\mc N_{q,d})$ given in Corollary~\ref{cor:upper bound private channel}, we have: 
\begin{corollary}\label{corollaryIII5}
    $\mc Q(\mc N_{q,d} \otimes \mc E_{\lambda,d}) > \mc Q(\mc N_{q,d}) + \mc Q(\mc E_{\lambda,d})$ if
\begin{align*}
    \log \left((d^2-1)(r_0 + |r_1|) +|r_0 + d r_1| + |r_1 + d r_0|\right) + \max\{(1-2\lambda) \log d , 0\} < 1 - \lambda h(q),
\end{align*}
with \begin{align*}
    r_0 = \frac{q}{d(d+1)} + \frac{1-q}{d(d-1)},\quad r_1 = \frac{q}{d(d+1)} - \frac{1-q}{d(d-1)}.
\end{align*}
\end{corollary}
As shown in Fig.~\ref{fig:corollaryIII5}, we see that even for $d=2$, one has $\mc Q(\mc N_{q,d} \otimes \mc E_{\lambda,d}) > \mc Q(\mc N_{q,d}) + \mc Q(\mc E_{\lambda,d})$ for some parameters $\lambda,q \in (0,1)$. In \cite{Smith_2008}, the least dimension of the erasure channel for superactivation is $d=3$. 

\begin{figure}[ht]
    \centering
    \includegraphics[width=1\linewidth]{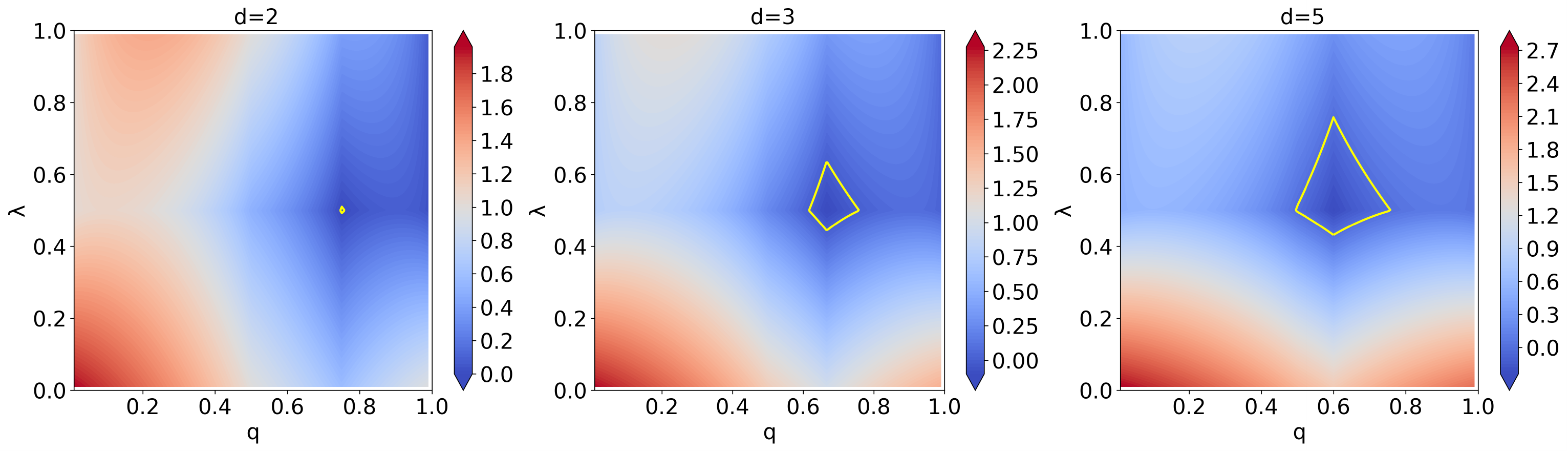}
    \caption{2D plots of the quantity (LHS – RHS) for the inequality $\log \left((d^2-1)(r_0 + |r_1|) + |r_0 + d r_1| + |r_1 + d r_0|\right) + \max\{(1-2\lambda)\log d, 0\} < 1 - \lambda h(q)$ in Corollary \ref{corollaryIII5}, where LHS and RHS denote the left- and right-hand sides of the inequality, respectively. The plots show (LHS – RHS) as functions of erasure channel parameter $\lambda$ and private channel parameter  $q$ for different dimension $d$, with the yellow solid line indicating the contour where $\text{LHS} - \text{RHS} = 0$.}
    \label{fig:corollaryIII5}
\end{figure}

\begin{figure}[ht]
    \centering
    \includegraphics[width=1\linewidth]{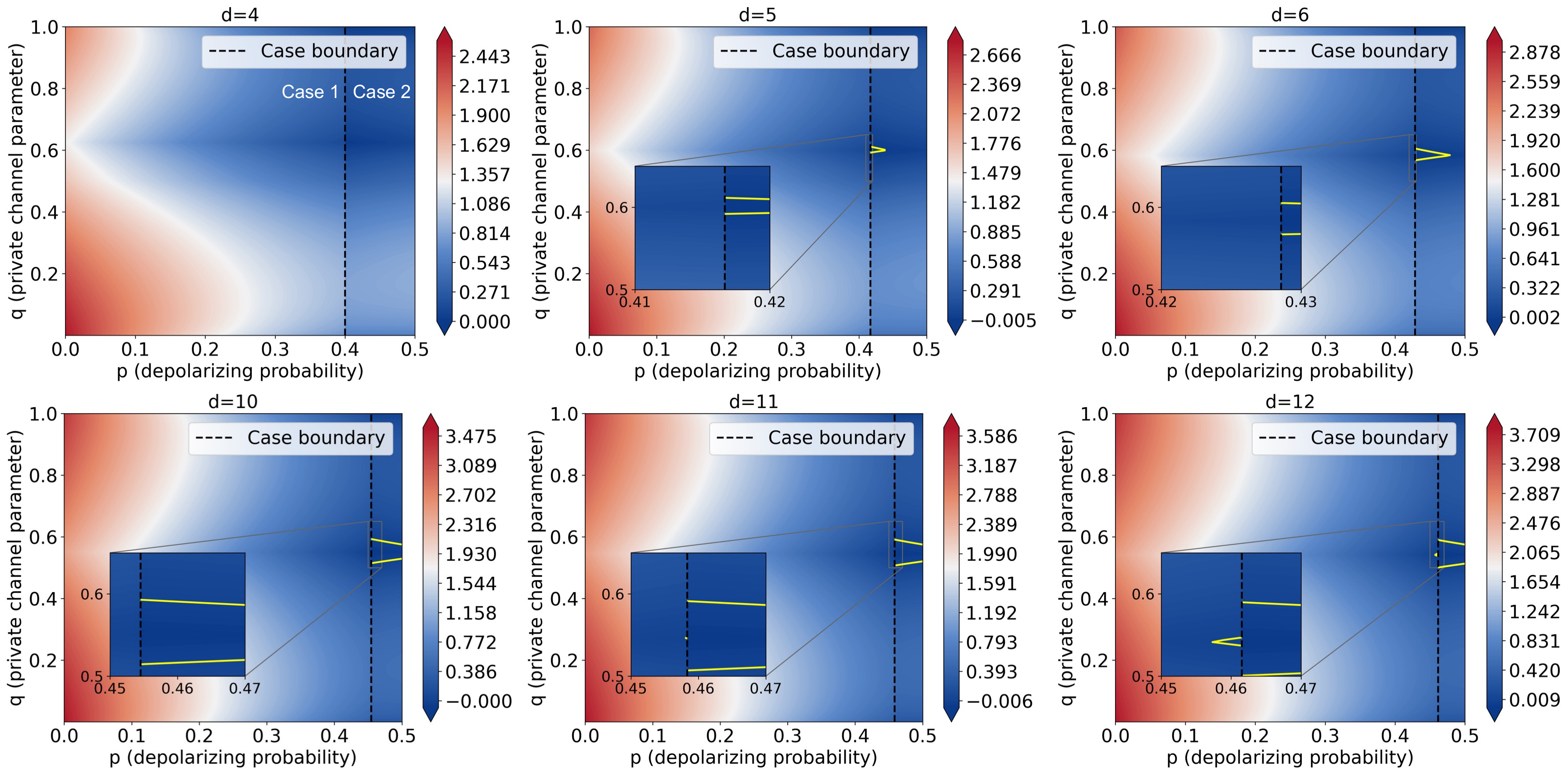}
    \caption{2D plots of the quantity (LHS – RHS) for the inequality in Corollary \ref{corollaryIII6}, where LHS and RHS denote the left- and right-hand sides of the inequality, respectively. The plots display (LHS – RHS) as functions of depolarizing probability $p$ and private channel parameter $q$ for different values of dimension $d$. The yellow solid line marks the contour where $\text{LHS} - \text{RHS} = 0$, while the two regimes are separated by the black dashed line indicating the case boundary.}
    \label{fig:corollaryIII6}
\end{figure}

\begin{comment}
    \begin{figure*}[ht]
\begin{center}
\includegraphics[width=0.7\columnwidth]{amplification_general.png}
\caption{sub-regions of $p,q,d$ with (super-)amplification} 
\end{center}
\end{figure*}

Here we briefly discuss the exact calculation of the quantum capacity of $\mc N^{A'\to B}$. Recall that $A = a_0A_0$, $B = b_0B_0$. The (normalized) Choi operator is given by
\begin{equation}
\begin{aligned}
       \mc J_{\mc N} & = \frac{1}{2d^2}(\ketbra{00}{00}^{a_0b_0} + \ketbra{11}{11}^{a_0b_0}) \otimes I_d^{A_0}\otimes I_d^{B_0} + \frac{1}{2d^2}(\ketbra{00}{11}^{a_0b_0} + \ketbra{11}{00}^{a_0b_0}) \otimes F^{A_0B_0} \\
       & = \frac{1}{2d^2} \sum_{i,j=0}^{d-1} \ketbra{0i0j}{0i0j}^{a_0A_0b_0B_0} + \ketbra{1i1j}{1i1j}^{a_0A_0b_0B_0} + \ketbra{0i0j}{1j1i}^{a_0A_0b_0B_0} + \ketbra{1i1j}{0j0i}^{a_0A_0b_0B_0} \\
       & = \frac{1}{d^2} \sum_{i,j=0}^{d-1} \ketbra{\Phi_{i,j}}{\Phi_{i,j}},
\end{aligned}
\end{equation}
where 
\begin{equation}
    \ket{\Phi_{i,j}}:= \frac{1}{\sqrt{2}} \left(\ket{0i0j} + \ket{1j1i}\right)^{a_0A_0b_0B_0}.
\end{equation}
Via the standard correspondence between Kraus operators and Choi operators, we have
\begin{equation}
    \mc N(\rho) = \sum_{i,j=0}^{d-1} K_{ij} \rho K_{ij}^{\dagger},\quad  K_{ij} = \frac{1}{\sqrt{d}} \left(\ket{0j}^{b_0B_0}\bra{0i}^{a_0A_0} + \ket{1i}^{b_0B_0}\bra{1j}^{a_0A_0} \right)
\end{equation}
\end{comment}

%--------------------------- Depolarizing---------------------------
\subsubsection*{Depolarizing channel}
When $\mc M = \mc D_{p,d}$, where $\mc D_{p,d}$ is the depolarizing channel defined by
\begin{equation}\label{depolarizing channel}
    \mc D_{p,d}(\rho):= (1-p) \rho + p \frac{\mb I_d}{d}.
\end{equation}
Then one can directly compute $\mc J_{\mc D_{p,d}}^{\sym}$ and $\mc J_{\mc D_{p,d}}^{\asym}$:
\begin{align*}
    & \mc J_{\mc D_{p,d}}^{\sym} = \frac{1}{2 d_{\sym}} [(1 + \frac{p}{d}) \mb I_d \otimes \mb I_d + (1-p) F] = \frac{1}{2 d_{\sym}} [(2 - p + \frac{p}{d}) P_{\sym} + (p + \frac{p}{d}) P_{\asym}],\\
    & \mc J_{\mc D_{p,d}}^{\asym} = \frac{1}{2 d_{\asym}} [(1 - \frac{p}{d}) \mb I_d \otimes \mb I_d - (1-p) F] = \frac{1}{2 d_{\asym}}[( p - \frac{p}{d}) P_{\sym} + (2 - p - \frac{p}{d}) P_{\asym}]
\end{align*}
Using the formula~\eqref{eqn:entropy orthogonal}, the entropy difference in~\eqref{sufficient:general} can be calculated by 
\begin{align*}
    & S(q \mc J^{\mathrm{sym}}_{\mc D_{p,d}} + (1-q)\mc J^{\mathrm{asym}}_{\mc D_{p,d}}) - qS(\mc J^{\mathrm{sym}}_{\mc D_{p,d}}) - (1-q) S(\mc J^{\mathrm{asym}}_{\mc D_{p,d}}) \\
    & = h\left((1-p)q + \frac{p}{2}(1+\frac{1}{d})\right) - q\cdot  h\left(1- \frac{p}{2}(1-\frac{1}{d})\right) - (1-q) \cdot h\left(\frac{p}{2}(1+\frac{1}{d})\right).
\end{align*}
Therefore, using Theorem \ref{thm:amplification}, a known upper bound on $\mc Q(\mc D_{p,d})$ \cite[Equation (10)]{Fanizza_2020}: 
\begin{equation}
    \mc Q(\mc D_{p,d}) \le \begin{cases}
        \log d + \eta(\frac{1}{2}) - \eta(\frac{1}{2} - \frac{d^2-1}{d^2}p) - (d^2-1)\eta(\frac{p}{d^2}),& p< \frac{d}{2(d+1)}, \\
        0, & p \ge \frac{d}{2(d+1)}
    \end{cases} 
\end{equation}
and the upper bound for $\mc Q(\mc N_{q,d})$ given in Corollary~\ref{cor:upper bound private channel}, we have: 
\begin{corollary}\label{corollaryIII6}
    $\mc Q(\mc N_{q,d} \otimes \mc D_{p,d}) > \mc Q(\mc N_{q,d}) + \mc Q(\mc D_{p,d})$ if: 
    \begin{itemize}
        \item Case 1: $p< \frac{d}{2(d+1)}$ and \begin{align*}
    & \log \left((d^2-1)(r_0 + |r_1|) +|r_0 + d r_1| + |r_1 + d r_0|\right) + \log d + \eta(\frac{1}{2}) - \eta(\frac{1}{2} - \frac{d^2-1}{d^2}p) - (d^2-1)\eta(\frac{p}{d^2}) \\
    & < 1 - h(q) + h\left((1-p)q + \frac{p}{2}(1+\frac{1}{d})\right) - q\cdot  h\left(1- \frac{p}{2}(1-\frac{1}{d})\right) - (1-q) \cdot h\left(\frac{p}{2}(1+\frac{1}{d})\right),\quad 
\end{align*}
with \begin{align*}
    r_0 = \frac{q}{d(d+1)} + \frac{1-q}{d(d-1)},\quad r_1 = \frac{q}{d(d+1)} - \frac{1-q}{d(d-1)}. 
\end{align*}
\item Case 2: $p\ge \frac{d}{2(d+1)}$ and 
\begin{align*}
    & \log \left((d^2-1)(r_0 + |r_1|) +|r_0 + d r_1| + |r_1 + d r_0|\right) \\
    & < 1 - h(q) + h\left((1-p)q + \frac{p}{2}(1+\frac{1}{d})\right) - q\cdot  h\left(1- \frac{p}{2}(1-\frac{1}{d})\right) - (1-q) \cdot h\left(\frac{p}{2}(1+\frac{1}{d})\right).
\end{align*}
    \end{itemize}
\end{corollary}
As shown in Fig.~\ref{fig:corollaryIII6}, the least dimension to see quantum capacity amplification is $d=5$ for the antidegradable region $p\ge \frac{d}{2(d+1)}$; for the region where $p\le \frac{d}{2d+2}$, the least dimension to see quantum capacity amplification is $d=11$.

\section{Gap between quantum capacity and private capacity}\label{sec:single letter}

In this section, we investigate the fundamental question of the separation between quantum and private capacities. It has been shown that such a separation can exist~\cite{Horodecki_2005}, and later that the gap can even be infinite~\cite{Leung_2014}, where the quantum capacity remains bounded by a constant while the private capacity diverges. Here, we provide an even stronger example. Our analysis begins with the derivation of the quantum capacity of $\mc N_{q,d}$ defined in~\eqref{eqn:channel}. Using the \emph{Spin Alignment Conjecture} (SAC), first proposed in~\cite{leditzky2023platypus} and reviewed in Conjecture~\ref{conj:SAC}, we show that the capacity is single-letter in a specific regime (that is, $\mc Q(\mc N)=\mc Q^{(1)}(\mc N)$), even though the channel is neither degradable nor anti-degradable. Building on this result, we construct channels whose quantum capacity vanishes while their private capacity diverges, demonstrating a sharper manifestation of the fundamental separation between the two capacities.
% We show that the quantum capacity of $\mc N_{q,d}$ defined in~\eqref{eqn:channel} \emph{is single-letter in a specific regime} (i.e., $\mc Q(\mc N)=\mc Q^{(1)}(\mc N)$ there), even though the channel is neither degradable nor anti-degradable. In fact, as demonstrated in Section~\ref{sec:amp}, it exhibits quantum-capacity amplification when paired with degradable channels. Intuitively, information encoded into the shield is leaked to the environment and cannot contribute to quantum communication; only the private subsystem can carry quantum communication. We formalize this intuition using the \emph{Spin Alignment Conjecture} (SAC) first proposed in~\cite{leditzky2023platypus} and we will review it in Conjecture~\ref{conj:SAC}.
% As an application, 
More specifically, we construct a family of channels $\{\mc M_n=\mc M_n^{A_n\to B_n}\}_{n\ge 1}$ such that
\begin{equation}
    \mc Q(\mc M_n)=\frac{1}{n}\to 0,\qquad \mc P(\mc M_n)=n\to \infty.
\end{equation}
\subsection{Single-letter quantum capacity using Spin alignment conjecture}
We choose a special $q = \frac{d+1}{2d}$ in \eqref{eqn:channel matrix general} and denote 
\begin{align*}
    \mc N:=\mc N_{\frac{d+1}{2d},d}.
\end{align*}
In this case, we have a simpler expression: for any density operator $\rho = \begin{pmatrix}
    X_{00} & X_{01} \\
    X_{10} & X_{11}
\end{pmatrix} \in \mc D(\mb C^2 \otimes \mb C^d)$, where $X_{rs} \in \mc B(\mb C^d),\ r,s\in \{0,1\}$, we have 
\begin{equation}\label{def:private channel 1}
    \mc N(\rho) = \frac{1}{d}\begin{pmatrix}
    \tr(X_{00})\mb I_d & X_{01}^T \\
    X_{10}^T & \tr(X_{11})\mb I_d \end{pmatrix}= \sum_{i,j=0}^{d-1} K_{ij}\, \rho \, K_{ij}^{\dagger},\  
\end{equation}
where $K_{ij} = \frac{1}{\sqrt{d}} \left(\ket{0j}^{b_0B_0}\bra{0i}^{a_0A_0} + \ket{1i}^{b_0B_0}\bra{1j}^{a_0A_0} \right)$.
\iffalse
In the matrix form, 
\begin{align*}
    \mc N(\rho) & = \sum_{i,j=0}^{d-1} K_{ij}\, \left(\sum_{r,s =0 }^1 \ketbra{r}{s}\otimes X_{rs}\right)\, K_{ij}^{\dagger} \\
    & = \frac{1}{d}\sum_{i,j=0}^{d-1}\bigg\{ (\ketbra{0}{0} \otimes \ketbra{j}{i})\, \left(\sum_{r,s =0 }^1 \ketbra{r}{s}\otimes X_{rs}\right)\, (\ketbra{0}{0} \otimes \ketbra{i}{j}) + (\ketbra{0}{0} \otimes \ketbra{j}{i})\, \left(\sum_{r,s =0 }^1 \ketbra{r}{s}\otimes X_{rs}\right)\, (\ketbra{1}{1} \otimes \ketbra{j}{i}) \\
    & + (\ketbra{1}{1} \otimes \ketbra{i}{j})\, \left(\sum_{r,s =0 }^1 \ketbra{r}{s}\otimes X_{rs}\right)\, (\ketbra{0}{0} \otimes \ketbra{i}{j}) + (\ketbra{1}{1} \otimes \ketbra{i}{j})\, \left(\sum_{r,s =0 }^1 \ketbra{r}{s}\otimes X_{rs}\right)\, (\ketbra{1}{1} \otimes \ketbra{j}{i})\bigg\} \\
    & = \frac{1}{d} \left(\ketbra{0}{0} \otimes \tr(X_{00})\mb I_d + \ketbra{0}{1} \otimes X_{01}^T+ \ketbra{1}{0} \otimes X_{10}^T + \ketbra{1}{1} \otimes \tr(X_{11})\mb I_d\right)
\end{align*}
where $A^T$ is the transpose of $A$. 
\fi
The complementary channel $\mc N^c: \mc D(\mb C^d \otimes \mb C^d) \to \mc D(\mb C^d \otimes \mb C^d)$ is given by 
\begin{equation}\label{eqn:complementary}
\begin{aligned}
    \mc N^c(\rho) & = \sum_{i,j,i'j' = 0}^{d-1} \tr(K_{ij} \rho K_{i'j'}^{\dagger}) \ketbra{ij}{i'j'} = \sum_{i,j,i'j' = 0}^{d-1} \tr(K_{i'j'}^{\dagger} K_{ij} \rho ) \ketbra{ij}{i'j'} \\
    & = \frac{1}{d}\sum_{i,j,i'j' = 0}^{d-1} \tr\left((\ketbra{0i'}{0j'} + \ketbra{1j'}{1i'}) (\ketbra{0j}{0i} + \ketbra{1i}{1j}) \rho \right) \ketbra{ij}{i'j'} \\
    & = \frac{1}{d}\sum_{i,j,i'j' = 0}^{d-1} \left(\delta_{jj'}\tr(\ketbra{0i'}{0i}\rho) + \delta_{ii'}\tr(\ketbra{1j'}{1j}\rho) \right) \ketbra{ij}{i'j'} \\
    & = \frac{1}{d}\sum_{i,j,i'j' = 0}^{d-1} \left(\delta_{jj'} \bra{i}X_{00} \ket{i'} + \delta_{ii'}\bra{j}X_{11} \ket{j'} \right) \ketbra{ij}{i'j'} \\
    & = \frac{1}{d} (X_{00} \otimes \mb I_d + \mb I_d \otimes X_{11}).
\end{aligned}
\end{equation}
Based on the calculation of $\mc N$ and $\mc N^c$, we can show the following:
\begin{lemma}
    For the quantum channel defined in \eqref{def:private channel 1}, we have 
    \begin{equation}
        \mc Q^{(1)}(\mc N) = \frac{1}{d}.
    \end{equation}
\end{lemma}
\begin{proof}
    Recall that $$\mc Q^{(1)}(\mc N) = \max_{\rho} S(\mc N(\rho)) -S(\mc N^c(\rho)) = \max_{\rho = \begin{pmatrix}
    X_{00} & X_{01} \\
    X_{10} & X_{11}
\end{pmatrix}} \left[S\left(\frac{1}{d}\begin{pmatrix}
    \tr(X_{00})\mb I_d & X_{01}^T \\
    X_{10}^T & \tr(X_{11})\mb I_d
\end{pmatrix} \right) - S\left(\frac{1}{d} (X_{00} \otimes \mb I_d + \mb I_d \otimes X_{11})\right) \right].$$
Since the second term does not involve $X_{01}, X_{10}$, thus we can assume $X_{01}= X_{10} = 0$, i.e., 
\begin{align}\label{key step 1:diagonalize}
    \mc Q^{(1)}(\mc N) = \max_{\rho = \begin{pmatrix}
    X_{00} & 0 \\
    0 & X_{11}
\end{pmatrix}} \left[S\left(\frac{1}{d}\begin{pmatrix}
    \tr(X_{00})\mb I_d & 0\\
    0 & \tr(X_{11})\mb I_d
\end{pmatrix}\right) - S\left(\frac{1}{d} (X_{00} \otimes \mb I_d + \mb I_d \otimes X_{11})\right) \right].
\end{align}
In fact, this follows from the majorization relation (see \cite[Problem II.5.5]{bhatia2013matrix})
\begin{equation}
    \begin{pmatrix}
    \tr(X_{00})\mb I_d & 0 \\
    0 & \tr(X_{11})\mb I_d
\end{pmatrix} \prec \begin{pmatrix}
    \tr(X_{00})\mb I_d & X_{01}^T \\
    X_{10}^T & \tr(X_{11})\mb I_d
\end{pmatrix},
\end{equation}
and via Schur concavity of von Neumann entropy, 
$$S\left(\frac{1}{d}\begin{pmatrix}
    \tr(X_{00})\mb I_d & X_{01}^T \\
    X_{10}^T & \tr(X_{11})\mb I_d
\end{pmatrix}\right) \le S\left(\frac{1}{d}\begin{pmatrix}
    \tr(X_{00})\mb I_d & 0 \\
    0 & \tr(X_{11})\mb I_d
\end{pmatrix}\right).$$
By definition of majorization, for any $X_{00},X_{11} \ge 0$ with $\tr(X_{00}) = p,\ \tr(X_{11}) = 1-p,\ p\in [0,1]$, we have
\begin{align*}
    X_{00} \otimes \mb I_d + \mb I_d \otimes X_{11} \prec p \ketbra{\psi_0}{\psi_0} \otimes \mb I_d + \mb I_d \otimes (1-p)\ketbra{\psi_1}{\psi_1},
\end{align*}
where $\ket{\psi_0},\ket{\psi_1}$ are arbitary pure states on $\mb C^d$, which implies 
\begin{equation}\label{key step 2:purify}
    S\left(\frac{1}{d} (X_{00} \otimes \mb I_d + \mb I_d \otimes X_{11})\right) \ge S\left(\frac{1}{d} (p \ketbra{\psi_0}{\psi_0} \otimes \mb I_d + \mb I_d \otimes (1-p)\ketbra{\psi_1}{\psi_1})\right).
\end{equation}
Therefore, via \eqref{key step 1:diagonalize} and \eqref{key step 2:purify}, we have 
\begin{align*}
     \mc Q^{(1)}(\mc N) &= \max_{\rho = \begin{pmatrix}
   p \ketbra{\psi_0}{\psi_0} & 0 \\
    0 & (1-p)\ketbra{\psi_1}{\psi_1}
\end{pmatrix}} \left[S\left(\frac{1}{d}\begin{pmatrix}
    p\mb I_d & 0\\
    0 & (1-p)\mb I_d
\end{pmatrix}\right) - S\left(\frac{1}{d} (p \ketbra{\psi_0}{\psi_0} \otimes \mb I_d + \mb I_d \otimes (1-p)\ketbra{\psi_1}{\psi_1}\right) \right] \\
& = \max_{p\in [0,1]} \left[S\left(\text{diag}\{\underbrace{\frac{p}{d},\cdots,\frac{p}{d}}_{d\text{\ many}},\underbrace{\frac{1-p}{d} \cdots, \frac{1-p}{d}}_{d\text{\ many}}\}\right) - S\left(\text{diag}\{\frac{1}{d}, \underbrace{\frac{p}{d},\cdots,\frac{p}{d}}_{d-1\text{\ many}},\underbrace{\frac{1-p}{d} \cdots, \frac{1-p}{d}}_{d-1\text{\ many}}\}\right) \right]\\
& = \max_{p\in [0,1]} \left[\log d + h(p) - (\log d + \frac{d-1}{d}h(p))\right] \\
& = \max_{p\in [0,1]} \frac{h(p)}{d}.
\end{align*}
Note that $h(p) \le 1$ with equality given by $p = \frac{1}{2}$, thus $\mc Q^{(1)}(\mc N) = \frac{1}{d}$.
\end{proof}
To calculate $\mc Q^{(1)}(\mc N^{\otimes n})$, we use the Spin Alignment Conjecture proposed in \cite{leditzky2023platypus}, and progress on resolving this conjecture can be seen in \cite{Alhejji_2025, Alhejji_2024}. Suppose $\sigma = \sum_{k=1}^d \lambda_k \ketbra{e_k}{e_k}$ is a density operator on $\mb C^d$ and $n \ge 1$. For each $M \subseteq \{1,2,\cdots,n\}$, let $M^c$ be the complement of $M$. We use $\om_{M} \ot \sigma^{\ot M^c}$ to denote a state on $(\mb C^d)^{\ot n}$ where
each subsystem labelled in $M^c$ is in the state $\sigma$, and the spins in $M$ are in a joint
state given by the density matrix $\om_{M}$. Let $\{x_M\}_{M\subseteq \{1,2,\cdots,n\}}$ be a probability distribution, that is, 
\begin{align}
    \sum_{M\subseteq \{1,2,\cdots,n\}} x_M = 1 \,, x_M \ge 0.
    \label{eq:x-distr}
\end{align}
The goal is to minimze the von Neuman entropy of $\kappa = \sum_{M} x_M   \om_{M} \ot \sigma^{\ot M^c}$, where $\om_M$ are variables (states). Formally, the {\bf entropy
minimization problem} is given by 
\begin{align}
    \label{spin alignment problem}
    & \min \{ S(\kappa): \kappa =  \sum_{M \subseteq \{1,2,\cdots, n\}} x_M   \om_{M} \ot \sigma^{\ot M^c},\quad  \om_{M}  \geq 0, \quad \Tr(\om_{M})= 1. \}
\end{align}

\begin{center}
\begin{tikzpicture}
    % Horizontal spacing
    \def\dx{2} 

    % Draw 6 spins (3 for M and 3 for \sigma)
    \node[circle, draw, minimum size=1cm] (M1) at (1*\dx,0) {\Large $\mb C^d$};
    \node[circle, draw, minimum size=1cm] (M2) at (2*\dx,0) {\Large $\mb C^d$};
    \node[circle, draw, minimum size=1cm] (M3) at (3*\dx,0) {\Large $\mb C^d$};
    \node[circle, draw, minimum size=1cm] (Q1) at (4*\dx,0) {\Large $\mb C^d$};
    \node[circle, draw, minimum size=1cm] (Q2) at (5*\dx,0) {\Large $\mb C^d$};
    \node[circle, draw, minimum size=1cm] (Q3) at (6*\dx,0) {\Large $\mb C^d$};

    % Connect spins with lines
    \draw[-] (M1.east) -- (M2.west);
    \draw[-] (M2.east) -- (M3.west);
    \draw[-] (M3.east) -- (Q1.west);
    \draw[-] (Q1.east) -- (Q2.west);
    \draw[-] (Q2.east) -- (Q3.west);

    % Labels for shared and individual states
    \node[above] at (M2.north) { $\omega_M$}; % Shared state
    \node[above] at (Q1.north) { $\sigma$};
    \node[above] at (Q2.north) { $\sigma$};
    \node[above] at (Q3.north) { $\sigma$};

    % Tensor product notation below
    %\node[below] at (M2.south) {\Huge $\omega_M \otimes \sigma^{M^c}$};
    % Add Underbrace for the "M" systems
    \node[below] at (2*\dx,-1) {$\underbrace{\hspace{5.2cm}}_{M}$};
    \node[below] at (5*\dx,-1) {$\underbrace{\hspace{5.2cm}}_{M^c}$};
\end{tikzpicture}
\end{center}

\begin{conjecture}[Spin Alignment Conjecture] \label{conj:SAC}
    For any fixed probability distribution $\{x_M\}_{M\subseteq \{1,2,\cdots,n\}}$, the entropy minimization problem in \eqref{spin alignment problem} is achieved at the state 
    \begin{equation}\label{eqn:spin alignment}
        \kappa = \sum_{M\subseteq \{1,2,\cdots,n\}} x_M   \ketbra{e_{k_0}}{e_{k_0}}^{\ot M} \ot \sigma^{\ot M^c},
    \end{equation}
    where $\ket{e_{k_0}}$ is the eigenvector corresponding to the maximal eigenvalue of $\sigma$. 
\end{conjecture}
\noindent Using the above conjecture, we are able to evaluate the quantum capacity of $\mc N$:
\begin{theorem}\label{main:quantum capacity}
    For the quantum channel defined in \eqref{def:private channel 1}, we have 
    \begin{equation}
       \mc Q^{(1)}(\mc N^{\otimes n}) = \frac{n}{d},\quad \forall n \ge 1.
    \end{equation}
    In particular, we have $\mc Q(\mc N) = \mc Q^{(1)}(\mc N)$.
\end{theorem}
\begin{proof}
    For any density operator $\rho^n \in \mc D((\mb C^2 \otimes \mb C^d)^{\otimes n})$, swapping the subsystems, we decompose it as a state in $\mc D((\mb C^2)^{\otimes n} \otimes (\mb C^d)^{\otimes n})$:
    \begin{equation}\label{eqn:n-qudit state}
        \rho^n = \sum_{\vec x, \vec y \in \{0,1\}^n} \ketbra{\vec x}{\vec y} \otimes X_{\vec x,\vec y},\quad X_{\vec x,\vec y} \in \mc B((\mb C^d)^{\otimes n}). 
    \end{equation}
    Here $\ketbra{\vec x}{\vec y} = \ketbra{x_1 x_2 \cdots x_n}{y_1y_2\cdots y_n} \in \mc B((\mb C^2)^{\otimes n})$. Using \eqref{def:private channel 1}, we can decompose $\mc N^{\otimes n}$ as super-operator acting on $\mc B((\mb C^2)^{\otimes n} \otimes (\mb C^d)^{\otimes n} )$:
    \begin{equation}
        \mc N^{\otimes n}(\rho^n) =\sum_{\vec x,\vec y\in \{0,1\}^n} \ketbra{\vec x}{\vec y} \otimes \bigotimes_{t =1}^n\mc N_{x_t,y_t}(X_{\vec x, \vec y}),
    \end{equation}
    where $$\mc N_{r,s}(X):= \begin{cases}
        \frac{1}{d}\tr(X)\mb I_d, \quad r = s, \\
        \frac{1}{d} X^T,\quad r\neq s.
    \end{cases}$$
To calculate $(\mc N^c)^{\otimes n}$ where $\mc N^c$ is defined via \eqref{eqn:complementary}, denote $\Pi_0, \Pi_1:  \mc B(\mb C^d) \to \mc B(\mb C^d \otimes \mb C^d)$ by
\begin{align}
    \Pi_0(X) = \frac{1}{d}X\otimes \mb I_d,\quad  \Pi_1(X) = \frac{1}{d}\mb I_d \otimes X.
\end{align}
Then we have 
$$(\mc N^c)^{\otimes n}(\rho^n) = \sum_{\vec x\in \{0,1\}^n} \bigotimes_{t = 1}^n \Pi_{x_t}(X_{\vec x, \vec x}).$$
As a result $\mc Q^{(1)}(\mc N^{\otimes n})$ is calculated by 
\begin{align*}
    \mc Q^{(1)}(\mc N^{\otimes n}) = \sup_{\rho^n\ \text{given\ by}\ \eqref{eqn:n-qudit state}} S\left(\sum_{\vec x,\vec y\in \{0,1\}^n} \ketbra{\vec x}{\vec y} \otimes \bigotimes_{t =1}^n\mc N_{x_t,y_t}(X_{\vec x, \vec y}) \right) - S\left( \sum_{\vec x\in \{0,1\}^n} \bigotimes_{t = 1}^n \Pi_{x_t}(X_{\vec x, \vec x})\right).
\end{align*}
Note that the first entropy involves non-diagonal operators $X_{\vec x, \vec y}$ with $\vec x \neq \vec y$ and the second entropy only involves diagonal operators $X_{\vec x,\vec x}$. Therefore, via majorization argument, the supremum is achieved at state $\rho^n$ with the block diagonal form:
\begin{align*}
    \rho^n = \sum_{\vec x\in \{0,1\}^n} \ketbra{\vec x}{\vec x} \otimes X_{\vec x,\vec x}. 
\end{align*}
Denote $p_{\vec x} = \tr(X_{\vec x,\vec x})$, we have 
\begin{align*}
    \mc Q^{(1)}(\mc N^{\otimes n}) = \sup \left\{S\left(\frac{1}{d^n}\sum_{\vec x\in \{0,1\}^n} p_{\vec x}\ketbra{\vec x}{\vec x} \otimes \mb I_{d^n} \right) - S\left(\sum_{\vec x\in \{0,1\}^n} p_{\vec x}\bigotimes_{t = 1}^n \Pi_{x_t}(X_{\vec x, \vec x}/p_{\vec x})\right)\right\}
\end{align*}
By spin alignment conjecture \eqref{eqn:spin alignment}, the minimum entropy for the complementary channel is
\begin{align*}
    S\left( \sum_{\vec x \in \{0,1\}^n} p_{\vec x} \bigotimes_{t = 1}^n (\Pi_{x_t} (\ketbra{0}{0}))\right)
\end{align*}
Therefore, we show that the optimizer of $\mc Q^{(1)}(\mc N^{\otimes n})$ is given by 
\begin{equation}
    \sum_{\vec x \in \{0,1\}^n} p_{\vec x} \ketbra{\vec x}{\vec x} \otimes \ketbra{0^n}{0^n}.
\end{equation}
Note that \begin{equation}
    \mc N\big|_{span \{\ket{00}, \ket{10}\}}
\end{equation}
is a degradable channel, thus we have 
    $$\mc Q^{(1)}(\mc N^{\otimes n}) = \mc Q^{(1)}(\mc N\big|_{span \{\ket{00}, \ket{10}\}}^{\otimes n}) = n\mc Q^{(1)}(\mc N\big|_{span \{\ket{00}, \ket{10}\}}) = \frac{n}{d},$$
    where the last equality follows from 
    \begin{align*}
        \frac{1}{d} = I_c(\mc N\big|_{span \{\ket{00}, \ket{10}\}}, \frac{1}{2}(\ketbra{00}{00} + \ketbra{10}{10})) \le \mc Q^{(1)}(\mc N\big|_{span \{\ket{00}, \ket{10}\}}) \le \mc Q^{(1)}(\mc N) = \frac{1}{d}.
    \end{align*}
\end{proof}

\subsection{Construction of a channel with arbitrarily large private capacity and arbitrarily small quantum capacity} \label{subsec:construction}
In this subsection, we exploit the Theorem \ref{main:quantum capacity} to construct a class of channels $\{\mc M_n = \mc M_n^{A_n \to B_n}\}_{n\ge 1}$ such that

\begin{equation}
    \mc Q(\mc M_n) = \frac{1}{n} \to 0,\quad \mc P(\mc M_n) = n \to \infty,
\end{equation}
which further strengths the extensiveness of quantum and private capacity. To show this result, we first show that the private capacity of the channel~\eqref{eqn:channel matrix general} 
$$\mc N_{\frac{2d}{d+1},d} =: \mc N$$ is one. Intuitively, this channel is induced by a pbit and the private system has dimension two, thus the capability of send classical information privately is one bit per use:
\begin{prop}\label{main:private capacity}
    The private capacity of $\mc N$ is $1$, independent of $d$.
\end{prop}
\begin{proof}
    First we show that $\mc P(\mc N) \ge \mc P^{(1)}(\mc N) \ge 1$. Recall that the private information of the channel \(\mathcal{N}\) for the ensemble \(\{p_x, \rho_x^A\}\) is defined as:
\begin{align*}
\mc P^{(1)}(\mc N) := \sup_{\{p_x, \rho_x^A\}} I_p(\{p_x, \rho_x^A\}, \mathcal{N}),\quad I_p(\{p_x, \rho_x^A\}, \mathcal{N}) := I(\mc X; B)  -  I(\mc X; E).
\end{align*}
We choose the ensemble of states $\{p_x,\rho_x\}_{x = 0,1}$ with $p_0 = p_1 = \frac{1}{2}$ and 
\begin{align*}
    \rho_0 = \ketbra{0}{0} \otimes \mb I_d /d,\quad \rho_1 = \ketbra{1}{1} \otimes \mb I_d /d.
\end{align*}
Using the expressions for $\mc N,\mc N^c$, see~\eqref{def:private channel 1} and \eqref{eqn:complementary}, it is straightforward to calculate 
\begin{align*}
    & I_p(\{p_x, \rho_x\}, \mathcal{N}) = I(\mc X;B) - I(\mc X;E) = S(B) - S(E) - (S(\mc X B) - S(\mc XE) ),\\
    & S(B) = S(E) = \log d,\quad S(\mc X B) = \log d,\quad S(\mc X E) = 1 + \log d,
\end{align*}
which implies a lower bound $\mc P(\mc N) \ge 1$. On the other hand, an SDP upper bound for classical capacity (thus also an upper bound for private capacity) of $\mc N^{A\to B}$ is given by $\mc C(\mc N) \le \log \beta(\mc N)$, see~\cite[Theorem 11]{Wang_2018}, where 
\begin{equation}
    \begin{aligned}
        & \beta(\mc N) = \min \tr(X^B) \\
        & \text{s.t.} \hspace{0.5cm} -R^{AB} \le \mc J_\mc N^{T_B} \le R^{AB} \\
        & \hspace{1cm}  -\mb I^A \otimes X^B \le (R^{AB})^{T_B} \le \mb I^A \otimes X^B.
    \end{aligned}
\end{equation}
We claim that $\beta(\mc N) \le 2$ thus we have $\mc P(\mc N) \le \mc C(\mc N) \le 1$. In fact, note that the unnormalized \Choi operator is
$$\mc J_\mc N = 2d\, \gamma^{a_0b_0A_0B_0}_{\frac{2d}{d+1}, d} = \frac{2}{d} \left( \ketbra{\psi_+}{\psi_+}^{a_0b_0} \otimes P_{\sym}^{A_0B_0} + \ketbra{\psi_-}{\psi_-}^{a_0b_0} \otimes P_{\asym}^{A_0B_0}\right),$$
then the elementary calculation shows that 
\begin{align*}
    \mc J_\mc N^{T_B} = \frac{1}{d} \left(\ketbra{00}{00}^{a_0b_0} + \ketbra{11}{11}^{a_0b_0} \right) \otimes \mb I^{A_0B_0} + \frac{1}{2} \left(\ketbra{01}{10}^{a_0b_0} + \ketbra{10}{01}^{a_0b_0} \right) \otimes \ketbra{\Psi_+}{\Psi_+}^{A_0B_0},
\end{align*}
where $\ket{\Psi^+} = \frac{1}{\sqrt d} \sum_{i=0}^{d-1} \ket{ii}$. Denote $P_0^{a_0b_0}$ as the projection onto the subspace $\mathrm{span}\{\ket{00},\ket{11}\}$ and $P_1^{a_0b_0} := \mb I^{a_0b_0} - P_0^{a_0b_0}$, we choose 
\begin{equation}
    R^{AB} = \frac{1}{d} P_0^{a_0b_0} \otimes \mb I^{A_0B_0} + \frac{1}{2} P_1^{a_0b_0} \otimes \ketbra{\Psi_+}{\Psi_+}^{A_0B_0}, \quad X^B = \frac{1}{d} \mb I^{B}.
\end{equation}
One can directly check that $-R^{AB} \le \mc J_\mc N^{T_B} \le R^{AB},\ -\mb I^A \otimes X^B \le (R^{AB})^{T_B} \le \mb I^A \otimes X^B$ and $\tr(X^B) = 2d\cdot \frac{1}{d} = 2$, which implies that $\beta(\mc N) \le 2$ and we conclude the proof that $\mc P(\mc N) = 1$. 
\end{proof}
\begin{corollary}
    Denote $\mc M_n = \mc N^{\otimes n}$ and $d = n^2$, we have 
    \begin{equation}
        \mc Q(\mc M_n) = \frac{1}{n},\quad \mc P(\mc M_n) = n.
    \end{equation}
\end{corollary}
\begin{proof}
    Using Theorem \ref{main:quantum capacity}, we have $\mc Q(\mc M_n) = \frac{n}{d} = \frac{1}{n}$ with $d = n^2$. On the other hand, via Proposition \ref{main:private capacity}, we have $\mc P(\mc M_n) = n \mc P(\mc N) = n$. 
\end{proof}
\begin{remark}
    Recall that in \cite{Leung_2014}, a general relation between private and quantum capacity is given by
    \begin{equation}
        \mc P(\mc N) \le \frac{1}{2}(\log d_A + \mc Q(\mc N))
    \end{equation}
    for any channel with input system $A$. In our construction, the input dimension is $\log d_A = n \log (2n^2) = n + 2n \log n$. Therefore, our example saturates the inequality up to a logarithmic order. Note that the example saturating the inequality in \cite{Leung_2014} has a constant lower bound for quantum capacity and in contrast, our example has quantum capacity approaching to zero. 
\end{remark}
%%%%%%%%%%%%----approximate private state---%%%%%%%%%%%%%%%%%%%
\section{Approximate private channel and its applications}\label{sec:approximate}
One of the important problems in quantum entanglement theory is whether there exists a bipartite state, such that it has no distillable entanglement but has positive distillable key \cite{Horodecki_2005}. This question was solved using the framework of approximate private states. To be more specific, it was shown in \cite[Theorem 7]{Horodecki_2009} that for any $\varepsilon>0$, there exists a pbit $\gamma^{a_0b_0A_0B_0}$ and a corresponding PPT state $\zeta^{a_0b_0A_0B_0}$ such that
    \begin{equation}\label{eqn:PPT 1}
        \|\zeta^{a_0b_0A_0B_0} - \gamma^{a_0b_0A_0B_0}\|_1 \le \varepsilon. 
    \end{equation}
The construction is give in \cite[proof of Theorem 7]{Horodecki_2009}. The following lemma is a generalization of \eqref{eqn:PPT 1}, which allows a PPT extension of approximate private states. It is presented in Lemma 2 in the supplementary material of \cite{Cubitt_2015}:
\begin{lemma}\label{lemma:PPT extension}
    For any $\varepsilon > 0$ and $N \geq 1$, there exists a pbit $\gamma^{a_0b_0A_0B_0}$ and a corresponding PPT state $\zeta^{a_0b_0 A^N_0 B^N_0}$ (that is,  $(\zeta^{a_0b_0 A^N_0 B^N_0})^{T_{b_0B_0^N}} \ge 0$), where $A_0^N:= A_0^{\otimes N}, B_0^N := B_0^{\otimes N}$ such that 
    \begin{align}
        \left\lVert \zeta^{a_0b_0 A_0 B_0} - \gamma^{a_0b_0A_0B_0} \right\rVert_1 \leq \varepsilon,
    \end{align}
    where $\zeta^{a_0b_0 A_0 B_0} := \tr_{A_0^{N-1} B_0^{N-1}} (\zeta^{a_0b_0 A^N_0 B^N_0})$ and it is independent of the choice of $N-1$ subsystems $A_0^{N-1} B_0^{N-1}$. 
\end{lemma}
For the readers' convenience, we present the construction of $\zeta$, see \cite[(139)]{Horodecki_2009} and \cite[(S18)]{Cubitt_2015}:
\begin{equation}\label{eqn:extension PPT state}
\zeta^{a_0b_0 A^N_0 B^N_0} = \zeta^{a_0b_0 A^N_0 B^N_0}_{q,d,r,N,m}\propto \left(\begin{array}{cccc}
\left[q\left(\frac{\tau_1+\tau_2}{2}\right)\right]^{\otimes m} & 0 & 0 & {\left[q\left(\frac{\tau_1-\tau_2}{2}\right)\right]^{\otimes m}} \\
0 & {\left[\left(\frac{1}{2}-q\right) \tau_2\right]^{\otimes m}} & 0 & 0 \\
0 & 0 & {\left[\left(\frac{1}{2}-q\right) \tau_2\right]^{\otimes m}} & 0 \\
{\left[q\left(\frac{\tau_1-\tau_2}{2}\right)\right]^{\otimes m}} & 0 & 0 & {\left[q\left(\frac{\tau_1+\tau_2}{2}\right)\right]^{\otimes m}}
\end{array}\right),
\end{equation}
where \begin{align*}
    \tau_1 = \left(\frac{1}{2}(\frac{P_{\sym}}{d_\sym} + \frac{P_\asym}{d_\asym})\right)^{\otimes rN},\quad \tau_2 = \left(\frac{P_{\sym}}{d_\sym}\right)^{\otimes rN}.
\end{align*}
In this case $A_0 = B_0= (\mb C^d)^{\otimes rm}$. Using the criteria developed in \cite{Horodecki_2009}, the state~\eqref{eqn:extension PPT state} is PPT, if 
\begin{align*}
    0< q \le \frac{1}{3},\quad \frac{1-q}{q} \ge \left(\frac{d}{d-1}\right)^{rN}. 
\end{align*}
To ensure \eqref{eqn:PPT 1} and PPT, one can set $q = \frac{1}{3}$, $r = \lceil2m + \log m \rceil$ and $m \cong \log(1/\varepsilon)$, with the freedom of choosing $d,N$ such that $ \left(\frac{d}{d-1}\right)^{rN} \le 2$. %In this case, since the reduced density of $\frac{P_{\sym}}{d_\sym}$ and $\frac{P_{\asym}}{d_\asym}$ is maximally mixed, it is easy to see that $\zeta^{a_0A_0}$ is maximally mixed state thus $\zeta^{a_0b_0 A^N_0 B^N_0}$ induces a quantum channel from $a'_0A_0'^N \to b_0B_0^N$.  

Motivated from the above question, in the remaining of this section, we explore different applications of approximate private channels defined as follows. 
\begin{definition}[Approximate private channel]
    The quantum channel $\mc M^{A'\to B}$ where $A' = a_0'A_0', B = b_0B_0$ is called an $\varepsilon$ approximate private channel, if its induced \Choi operator is $\varepsilon$ close to some pbit $\gamma^{a_0b_0A_0B_0}$ in trace distance. 
\end{definition}
%---------------------------------Application 1-------------------------
\subsection{Application 1: A different proof of superactivation effect}
\begin{prop}\label{prop:amplification 2}
    Suppose the channel $\mc M^{A' \to B}$ with \Choi operator $\zeta^{a_0b_0A_0B_0}$ is an $\varepsilon$-approximate private channel, i.e., $\|\zeta^{a_0b_0A_0B_0} - \gamma^{a_0b_0A_0B_0}\|_1 \le \varepsilon$ for some pbit $\gamma^{a_0b_0A_0B_0}$. Then for any $\lambda \in (0,1)$, 
\begin{align} 
\mc Q^{(1)} \left(\mc M^{A' \to B} \otimes \mc E_{\lambda,d}\right) \ge 1 - \lambda\, h(\frac{1+|c|}{2})- 4\varepsilon - 2h(\varepsilon),
\end{align}
where $c = \tr(\bra{00}^{a_0b_0} \gamma^{a_0b_0A_0B_0} \ket{11}^{a_0b_0})$. 
\end{prop}
\begin{proof}
    Note that $\mc Q^{(1)}(\mc M^{A' \to B} \otimes \mc E_{\lambda,d})  \ge I_c(\mc M^{A' \to B} \otimes \mc E_{\lambda,d}, \rho^{A'A_0})$, where the ansatz state $\rho^{A'A_0}$ is given by 
    \begin{align*}
        \rho^{A'A_0} = \tr_{a_0} \left(\ketbra{\psi}{\psi}^{a_0a_0'} \otimes \ketbra{\Psi}{\Psi}^{A_0A_0'}\right). 
    \end{align*}
    Then the coherent information is calculated by 
    \begin{align*}
        I_c(\mc M^{A' \to B} \otimes \mc E_{\lambda,d}, \rho^{A'A_0}) & = (1-\lambda) I_c(\mc M^{A' \to B} \otimes id^{A_0\to A_0}, \rho^{A'A_0}) + \lambda I_c(\mc M^{A' \to B} \otimes \mc E_1^{A_0\to C_0}, \rho^{A'A_0}) \\
        &= (1-\lambda) I(a_0 \rangle b_0A_0B_0)_{\zeta} + \lambda I(a_0 \rangle b_0B_0)_{\zeta^{a_0b_0B_0}} \\
        & \ge (1-\lambda) I(a_0 \rangle b_0A_0B_0)_{\zeta} + \lambda I(a_0 \rangle b_0)_{\zeta^{a_0b_0}}.
    \end{align*}
    For the first equality above, we use the fact that coherent information is additive under flag mixture of two channels~\eqref{coh:flag}; for the second equality, we use the definition of coherent information; for the last inequality, we use the data processing inequality for coherent information~\eqref{DPI:coherent}. Using the well-known continuity of conditional entropy (Fannes-Alicki inequality~\cite{fannes1973continuity}) and the fact that conditional entropy is the minus coherent information, we have 
\begin{align*}
    |I(a_0 \rangle b_0A_0B_0)_{\zeta} - I(a_0 \rangle b_0A_0B_0)_{\gamma}| \le 4\varepsilon + 2h(\varepsilon),\quad |I(a_0 \rangle b_0)_{\zeta^{a_0b_0}} - I(a_0 \rangle b_0)_{\gamma^{a_0b_0}}| \le 4\varepsilon + 2h(\varepsilon),
\end{align*}
we conclude the proof by applying Lemma~\ref{pbit:basic}. 
\end{proof}
Therefore, for fixed $\lambda \in [\frac{1}{2},1)$ and $\varepsilon$ small enough, the above lower bound is strictly positive, thus providing a simpler proof of superactivation effect~\cite{Smith_2008} when the approximate private channel is PPT which has zero quantum capacity. Note that we can slightly improve the bound $4\varepsilon + 2h(\varepsilon)$ using the recent results in \cite{berta2024continuity,audenaert2024continuity}.

%--------------------------- Application 2 -------------------------------------
\subsection{Application 2: Separation between anti-degradable channels and approximate private channels}
\label{subsec:sep-AD}

In Prop.~\ref{prop:amplification 2}, we show that tensoring an $\varepsilon$-approximate private channel with an anti-degradable erasure channel leads to superactivation. To demonstrate that this phenomenon cannot be attributed merely to approximate degradability, we further establish that $\varepsilon$-approximate private channels are quantitatively far, in the diamond norm, from the set of anti-degradable channels.

\begin{definition}[Anti-degradable channel]
A channel $\mc A^{A'\to B}$ is \emph{anti-degradable} if there exists a CPTP map $\mc T$ such that $\mc A=\mc T\circ \mc A^{c}$, where $\mc A^{c}$ is a complementary channel of $\mc A$. Denote the set of all such channels by $\mathrm{AD}(A'  \to  B)$.
\end{definition}
The key idea is to invoke the continuity of quantum capacities, shown in~\cite[Corollary 14]{Leung_2009}. 
\begin{lemma}\label{lemma:continuity of capacity}
The quantum capacity of a quantum channel with finite-dimensional output is continuous.
Quantitatively, if $\mathcal N,\mathcal M: A' \to B$ where the dimension of $B$ is $d_B$
and $\|\mathcal N-\mathcal M\|_{\diamond} \le \varepsilon$, then
\begin{equation}\label{eq:Q-cont}
\bigl|\mc Q(\mathcal N) -\mc Q(\mathcal M)\,\bigr|  \le  8\varepsilon\,\log d_B  +  4\,h(\varepsilon).
\end{equation}
\end{lemma}
\begin{theorem}[Quantitative separation from anti-degradable channels]
\label{thm:sep-AD}
Let $\mc M^{A'\to B}$ be an $\varepsilon$-approximate private channel with \Choi operator $\zeta^{a_0b_0A_0B_0}$ satisfying $\|\zeta-\gamma\|_1\le \varepsilon$ for some pbit $\gamma^{a_0b_0A_0B_0}$. Let $d = \mathrm{dim} A_0 = \mathrm{dim} B_0$ and 
\[
c :=\tr\Bigl(\bra{00}^{a_0b_0}\,\gamma^{a_0b_0A_0B_0}\,\ket{11}^{a_0b_0}\Bigr),
\quad
\Delta(\lambda,\varepsilon,c):=1-\lambda\,h\Bigl(\tfrac{1+|c|}{2}\Bigr)-4\varepsilon-2h(\varepsilon).
\]
Then for any $\lambda\in[\tfrac12,1)$, we have 
\begin{equation}
\inf_{\mc A\in \mathrm{AD}(A'  \to  B)} \,\|\mc M-\mc A\|_\diamond
\ge \min \left\{ g^{-1} (\Delta(\lambda,\varepsilon,c)), \frac{1}{2} \right\},
\end{equation}
where $g(x):= 8x\log(2d(d+1)) + 4h(x)$ is a strictly increasing function on $[0,\frac{1}{2}]$. 
\end{theorem}
\begin{proof}
First note that if $\mc A^{A_1\to B_1}\in \mathrm{AD}(A_1\to B_1)$ and $\mc F^{A_2 \to B_2}\in \mathrm{AD}(A_2 \to B_2)$, then $\mc A\otimes \mc F\in \mathrm{AD}(A_1A_2 \to B_1B_2)$. Therefore, via Proposition \ref{prop:amplification 2}, we have 
\begin{align*}
    \Delta(\lambda,\varepsilon,c) \le \mc Q^{(1)} (\mc M \otimes \mc E_{\lambda,d})= \mc Q^{(1)} (\mc M \otimes \mc E_{\lambda,d}) - \mc Q^{(1)} (\mc A \otimes \mc E_{\lambda,d})
\end{align*}
for any anti-degradable channel $\mc A = \mc A^{A'\to B}$. By Lemma~\ref{lemma:continuity of capacity}, we have 
\begin{align*}
     |\mc Q(\mc M \otimes \mc E_{\lambda,d}) - \mc Q(\mc A \otimes \mc E_{\lambda,d})| & \le 8 \|\mc M \otimes \mc E_{\lambda,d}- \mc A\otimes \mc E_{\lambda,d}\|_{\diamond} \log(2d(d+1)) + 4 h(\|\mc M \otimes \mc E_{\lambda,d}- \mc A\otimes \mc E_{\lambda,d}\|_{\diamond}) \\
     & \le  8 \|\mc M- \mc A\|_{\diamond} \log(2d(d+1)) + 4 h(\|\mc M- \mc A\|_{\diamond}),
\end{align*}
where the last inequality follows from $\|\mc M \otimes \mc E_{\lambda,d}- \mc A\otimes \mc E_{\lambda,d}\|_{\diamond} \le \|\mc M- \mc A\|_{\diamond}$ and $g(x):= 8x\log(2d(d+1)) + 4h(x)$ is a strictly increasing function on $[0,\frac{1}{2}]$. Therefore, assuming $\|\mc M- \mc A\|_{\diamond} \le \frac{1}{2}$, we conclude the proof. 
\end{proof}
A useful upper bound for the quantum capacity of a general channel is to reduce to
``nearby'' tractable families, e.g.\ degradable or anti-degradable channels.
For degradable proximity, several works bound $\mc Q(\mc N)$ in terms of an
\emph{approximate degradability parameter} $\delta_{\mathrm{deg}}(\mc N)$
(see, e.g.,~\cite{Sutter_2017,leditzky2017useful,Zhu_2024,Zhu_2025}).
However, our separation shows that this program can fail for
$\varepsilon$-approximate private channels.

%--------------------------- Application 3 -------------------------------------

\subsection{Application 3: Quantum capacity detection at arbitrarily large level}\label{sec:many copy}
In this section we apply the technique developed above to prove the following.  
For every integer $n\ge 1$ there exists a quantum channel
$\mathcal{N}_n^{A_n\to B_n}$ such that
\begin{equation}\label{eqn:threshold}
  \mathcal{Q}^{(1)}  \big(\mathcal{N}_n^{\otimes n}\big)=0
  \quad\text{while}\quad
  \mathcal{Q}\big(\mathcal{N}_n\big)>0.
\end{equation}
The phenomenon in \eqref{eqn:threshold} was first observed by Cubitt et al.~\cite{Cubitt_2015}; our construction is more explicit and arguably simpler.

\medskip
\begin{remark}
    The result above \textbf{does not} resolve the computability of the (regularized) quantum capacity
\(
  \mathcal{Q}(\mathcal{N})=\lim_{k\to \infty}\frac{1}{k}\,\mathcal{Q}^{(1)}(\mathcal{N}^{\otimes k}).
\)
In principle, the above channel $\mathcal{N}$ could satisfy
\(
  \mathcal{Q}(\mathcal{N})=\frac{1}{N}\,\mathcal{Q}^{(1)}  \big(\mathcal{N}^{\otimes N}\big)
\)
for some finite $N\ge 1$; our construction or the result in \cite{Cubitt_2015} neither rules this out nor provides an effective bound on such a $N$.
\end{remark}
\begin{figure}
    \centering
    \includegraphics[width=0.5\linewidth]{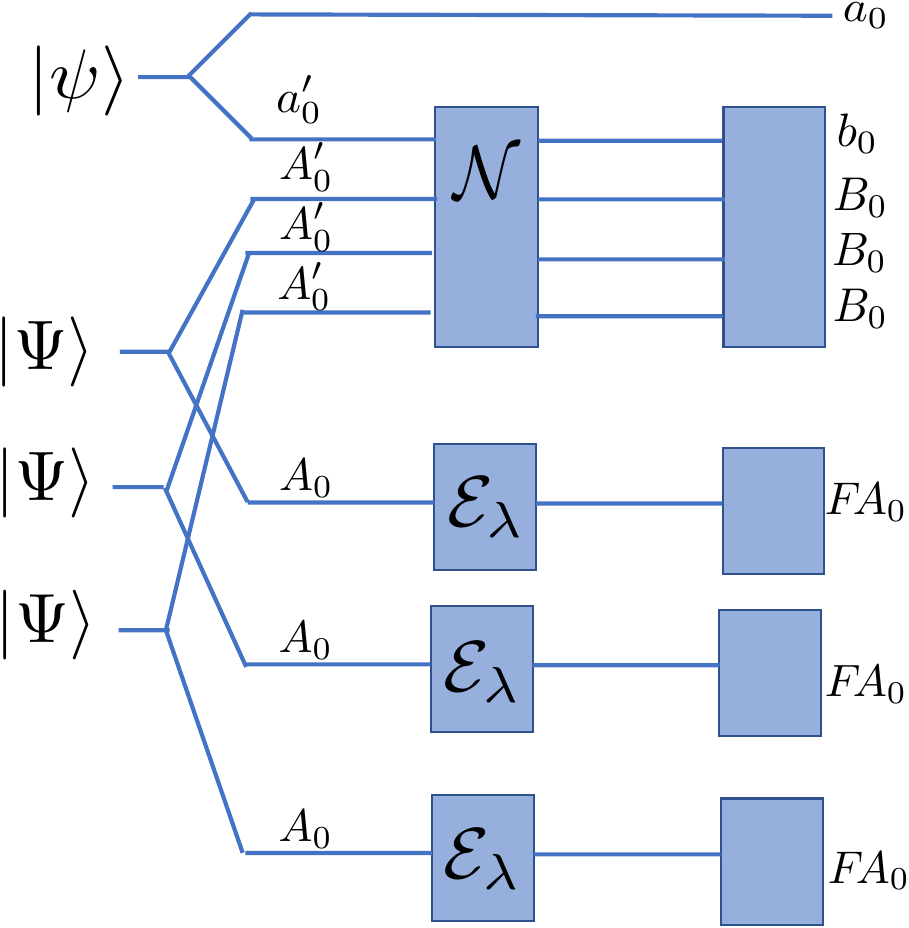}
    \caption{Superactivation with the help of many copy use of channels}
    \label{fig:amplification 2}
\end{figure}

A first step is a $N$-copy generalization of Proposition \ref{prop:amplification 2}:
\begin{prop}\label{prop:amplification 3}
For any $N\ge 1$ and $\varepsilon \in (0,1)$, suppose the state $\zeta^{a_0b_0 A^N_0 B^N_0}$ induces a quantum channel $\Gamma^{a_0' A_0'^N \to b_0B_0^N}$ and satisfies
\begin{align*}
    \left\lVert \zeta^{a_0b_0 A_0 B_0} - \gamma^{a_0b_0A_0B_0} \right\rVert_1 \leq \varepsilon
\end{align*}
for some pbit $\gamma^{a_0b_0A_0B_0}$. Then for any $\lambda \in (0,1)$, 
\begin{align} 
\mc Q^{(1)} \left(\Gamma^{a_0' A_0'^N \to b_0B_0^N} \otimes (\mc E_{\lambda}^{A_0 \to FA_0})^{\otimes N} \right) \ge 1 - \lambda^N h(\frac{1+|c|}{2})- 4\varepsilon - 2h(\varepsilon),
\end{align}
where $c = \tr(\bra{00}^{a_0b_0} \gamma^{a_0b_0A_0B_0} \ket{11}^{a_0b_0})$. 
\end{prop}
\begin{proof}
Denote $$A_0^N = A_{0,1}\cdots A_{0,N},\quad A_0'^N = A'_{0,1}\cdots A'_{0,N},\quad A' = a_0' A_0'^N$$
where each $A_{0,i}$ is isomorphic to $A_0$ labeled by $i$. The \Choi operator of $\Gamma^{A' \to b_0B_0^N}$ is given by
\begin{align*}
   \zeta^{a_0b_0 A^N_0 B^N_0} = \left( id^{a_0 A_0^N\to a_0 A_0^N} \otimes \Gamma^{A' \to b_0B_0^N} \right) \left(\ketbra{\psi}{\psi}^{a_0a_0'} \otimes \ketbra{\Psi}{\Psi}^{A_{0,1}A_{0,1}'^1} \otimes \cdots \otimes \ketbra{\Psi}{\Psi}^{A_{0,N}A_{0,N}'}\right).
\end{align*}
Denote $\wt B_0 := F A_0 \cong A_0 \oplus A_0$ and
\begin{equation}
    \rho^{a_0b_0B_0^N \wt B_0^N} := \left(id^{a_0 \to a_0} \otimes \Gamma^{A' \to b_0B_0^N} \otimes (\mc E_{\lambda}^{A_0 \to \wt B_0})^{\otimes N} \right)  \left(\ketbra{\psi}{\psi}^{a_0a_0'} \otimes \ketbra{\Psi}{\Psi}^{A_{0,1}A_{0,1}'^1} \otimes \cdots \otimes \ketbra{\Psi}{\Psi}^{A_{0,N}A_{0,N}'} \right),
\end{equation}
Choose the input state as 
\begin{align}\label{ansatz state}
    \sigma^{A'A_0^N} = \tr_{a_0}\left(\ketbra{\psi}{\psi}^{a_0a_0'} \otimes \ketbra{\Psi}{\Psi}^{A_{0,1}A_{0,1}'^1} \otimes \cdots \otimes \ketbra{\Psi}{\Psi}^{A_{0,N}A_{0,N}'}\right),
\end{align}
then via Lemma \ref{lemma:basic} for flagged channels, 
\begin{equation}\label{N copy:e1}
\begin{aligned}
    & I(a_0 \rangle b_0B_0^N \wt B_0^N)_{\rho^{a_0b_0B_0^N \wt B_0^N}} = I_c(\Gamma^{A' \to b_0B_0^N} \otimes (\mc E_{\lambda}^{A_0 \to \wt B_0})^{\otimes N}, \sigma^{A'A_0^N} ) \\
    & = \sum_{\textbf{b} \in \{0,1\}^N} (1-\lambda)^{\sum_{i=1}^N b_i} \lambda^{N - \sum_{i=1}^N b_i}I_c(\Gamma^{A' \to b_0B_0^N} \otimes \bigotimes_{i=1}^N \mc N_{b_i}, \sigma^{A'A_0^N}),
\end{aligned}
\end{equation}
where $\mc N_0 = \mc E_1^{A_0 \to A_0}$ and $\mc N_1 = id^{A_0 \to A_0}$. Given $\textbf{b} \in \{0,1\}^N$, we have two cases: (i) $\sum_{i=1}^N b_i = 0$; (ii) $\sum_{i=1}^N b_i>0$. The calculation proceeds as follows: 
\begin{align*}
    \text{(i)} \sum_{i=1}^N b_i = 0: & \quad \lambda^N I_c(\Gamma^{A' \to b_0B_0^N} \otimes \bigotimes_{i=1}^N \mc E_1^{A_{0,i}\to A_{0,i}}, \sigma^{A'A_0^N} ) \\
    & = \lambda^N I(a_0 \rangle b_0 B_0^NA_0^N)_{\rho^{a_0b_0B_0^N} \otimes  \mu^{\otimes N}}  = \lambda^N I(a_0 \rangle b_0 B_0^N)_{\rho^{a_0b_0B_0^N}} \ge \lambda^N I(a_0 \rangle b_0)_{\zeta^{a_0b_0}},
\end{align*}
where the above equalities follow from the definition of coherent information; the inequality follows from data processing inequality \eqref{DPI:remove}. For the other case,
\begin{align*}
    \text{(ii)} \sum_{i=1}^N b_i > 0: & \quad \sum_{\substack{\exists 1 \le j \le N\\ b_j = 1}} (1-\lambda)^{\sum_{i=1}^N b_i} \lambda^{N - \sum_{i=1}^N b_i} I_c(\Gamma^{A' \to b_0B_0^N} \otimes \bigotimes_{i=1}^N \mc N_{b_i}^{A_{0,i}\to A_{0,i}}, \sigma^{A'A_0^N} ) \\
    & \ge \sum_{\substack{\exists 1 \le j \le N\\ b_j = 1}} (1-\lambda)^{\sum_{i=1}^N b_i} \lambda^{N - \sum_{i=1}^N b_i} I_c(\Gamma^{A' \to b_0B_0^N} \otimes id^{A_{0,j}\to A_{0,j}}, \sigma^{A'A_{0,j}} ) \\
    & = \sum_{\substack{\exists 1 \le j \le N\\ b_j = 1}} (1-\lambda)^{\sum_{i=1}^N b_i} \lambda^{N - \sum_{i=1}^N b_i} I(a_0\rangle b_0 A_{0,j}B_0^N)_{\zeta^{a_0b_0A_{0,j} B_0^N}} \\
    & \ge \sum_{\substack{\exists 1 \le j \le N\\ b_j = 1}} (1-\lambda)^{\sum_{i=1}^N b_i} \lambda^{N - \sum_{i=1}^N b_i} I(a_0\rangle b_0 A_{0,j}B_{0,j})_{\zeta^{a_0b_0A_{0,j}B_{0,j}}} \\
    & = \sum_{\substack{\exists 1 \le j \le N\\ b_j = 1}} (1-\lambda)^{\sum_{i=1}^N b_i} \lambda^{N - \sum_{i=1}^N b_i} I(a_0\rangle b_0 A_{0}B_{0})_{\zeta^{a_0b_0A_{0}B_{0}}} = (1-\lambda^N) I(a_0\rangle b_0 A_{0}B_{0})_{\zeta^{a_0b_0A_{0}B_{0}}}.
\end{align*}
For the first inequality above, note that at least one register has perfect channel $\mc N_1$, then $I_c(\Gamma^{A' \to b_0B_0^N} \otimes \bigotimes_{i=1}^N \mc N_{b_i}^{A_{0,i}\to A_{0,i}}, \sigma^{A'A_0^N} ) \ge I_c(\Gamma^{A' \to b_0B_0^N} \otimes id^{A_{0,j}\to A_{0,j}}, \sigma^{A'A_{0,j}} )$ follows from data processing inequality if there are more than one register with perfect channel; otherwise we have equality since the register has a completely erasure map and we can apply \eqref{DPI:equality}. The second inequality follows from data processing inequality~\eqref{DPI:remove}.  

Using the same argument as in the proof of Proposition \ref{prop:amplification 2}, we have \begin{align*}
    |I(a_0 \rangle b_0A_0B_0)_{\zeta} - I(a_0 \rangle b_0A_0B_0)_{\gamma}| \le 4\varepsilon + 2h(\varepsilon),\quad |I(a_0 \rangle b_0)_{\zeta^{a_0b_0}} - I(a_0 \rangle b_0)_{\gamma^{a_0b_0}}| \le 4\varepsilon + 2h(\varepsilon),
\end{align*}
Combining the above two estimates and plugging them into \eqref{N copy:e1}, we have 
\begin{align*}
    \mc Q^{(1)} \left(\Gamma^{a_0' A_0'^N \to b_0B_0^N} \otimes (\mc E_{\lambda}^{A_0 \to FA_0})^{\otimes N} \right) & \ge I(a_0 \rangle b_0B_0^N \wt B_0^N)_{\rho^{a_0b_0B_0^N \wt B_0^N}} \\
    & \ge \lambda^N \left((1- h(\frac{1+|c|}{2}))- (2 h(\varepsilon) + 4\varepsilon)\right) + (1-\lambda^N)\left(1 - (2 h(\varepsilon) + 4\varepsilon)\right)\\
    & = 1 - \lambda^N h(\frac{1+|c|}{2})- 4\varepsilon - 2h(\varepsilon).
\end{align*}
\end{proof}
\noindent Using Proposition \ref{prop:amplification 3}, we can prove the following super-activation result:
\begin{theorem}\label{main:non-additivity}
    Given $N\ge 1$ and denote $A' = a_0 A_0^N, B = b_0B_0^N$. Suppose $\Gamma_\kappa^{A' \to F B}$ is given by
    \begin{align*}
        \Gamma_\kappa^{A' \to F B} = (1-\kappa) \ketbra{0}{0}^F\otimes \Gamma^{A' \to B} + \kappa \ketbra{1}{1}^F \otimes \mc E_1^{A' \to B},\quad \kappa\in (0,\frac{1}{2})
    \end{align*}
    where $\Gamma^{A' \to B}$ is induced by the state $\zeta^{AB} = \zeta^{a_0b_0A_0^NB_0^N}$, with $\|\zeta^{a_0b_0A_0B_0} - \gamma^{a_0b_0A_0B_0}\|_1 \le \varepsilon$ for some pbit $\gamma^{a_0b_0A_0B_0}$. Then for \begin{equation}
        N > \frac{\log(1 - 2h(\varepsilon) - 4\varepsilon - \frac{\kappa}{1-\kappa}) - \log(h(\frac{1+|c|}{2}))}{\log \lambda},
    \end{equation}
    where $c = \tr(\bra{00}^{a_0b_0} \gamma^{a_0b_0A_0B_0} \ket{11}^{a_0b_0})$, we have 
    \begin{align*}
        \mc Q^{(1)}\left((\Gamma_\kappa^{A' \to F B} \oplus \mc E_{\lambda}^{A_0 \to FA_0})^{\otimes N}\right) > 0. 
    \end{align*}
\end{theorem}
\begin{proof}
    Using the same input state $\sigma^{A'A_0^N}$~\eqref{ansatz state}, and applying Lemma \ref{lemma:basic}, we have 
    \begin{align*}
           & \mc Q^{(1)}\left((\Gamma_\kappa^{A' \to F B} \oplus \mc E_{\lambda}^{A_0 \to FA_0})^{\otimes N}\right) \ge \mc Q^{(1)}\left(\Gamma_\kappa^{A' \to F B} \otimes (\mc E_{\lambda}^{A_0 \to FA_0})^{\otimes N}\right) \\
           & \ge I_c(\Gamma_\kappa^{A' \to FB} \otimes (\mc E_{\lambda}^{A_0 \to FA_0})^{\otimes N}, \sigma^{A'A_0^N} ) \\
           & = (1-\kappa)I_c(\Gamma^{A' \to B} \otimes (\mc E_{\lambda}^{A_0 \to FA_0})^{\otimes N}, \sigma^{A'A_0^N} ) + \kappa I_c(\mc E_1^{A' \to B} \otimes (\mc E_{\lambda}^{A_0 \to FA_0})^{\otimes N}, \sigma^{A'A_0^N}).
    \end{align*}
    To proceed, for the first quantity above, we apply Proposition \ref{prop:amplification 3}; for the second quantity, we use the worse case bound $I_c(\mc E_1^{A' \to B} \otimes (\mc E_{\lambda}^{A_0 \to FA_0})^{\otimes N}, \sigma^{A'A_0^N}) \ge I_c(\mc E_1^{A'A_0^N \to A'A_0^N}, \sigma^{A'A_0^N}) = -S(\sigma^{A'A_0^N}) = -\log d_{a_0} = -1$. Therefore, we have 
    \begin{align*}
        \mc Q^{(1)}\left((\Gamma_\kappa^{A' \to F B} \oplus \mc E_{\lambda}^{A_0 \to FA_0})^{\otimes N}\right)\ge (1-\kappa) \left(1 - \lambda^N h(\frac{1+|c|}{2})- 4\varepsilon - 2h(\varepsilon)\right) -\kappa. 
    \end{align*}
    Let the above quantity be greater than zero, we get 
    \begin{align*}
        1 - \lambda^N h(\frac{1+|c|}{2})- 4\varepsilon - 2h(\varepsilon) > \frac{\kappa}{1-\kappa} \iff  N > \frac{\log(1 - 2h(\varepsilon) - 4\varepsilon - \frac{\kappa}{1-\kappa}) - \log(h(\frac{1+|c|}{2}))}{\log \lambda}.
    \end{align*}
\end{proof}
To ensure additivity up to any level, we use the following flag trick:
\begin{prop}\label{main:additivity}
    Suppose $\Gamma^{A_1\to B_1}$ is a zero-capacity channel. Then for any $\kappa \in (0,1)$ and $n\ge 1$, there exists $\lambda = \lambda(n,\kappa) <1$ such that $$\mc Q^{(1)}\bigg((\Gamma_\kappa^{A_1 \to F_1B_1} \oplus \mc E_{\lambda}^{A_2 \to F_2B_2})^{\otimes n}\bigg) = 0,$$ where
    \begin{equation}
        \Gamma_\kappa^{A_1 \to F_1B_1} := (1-\kappa)\ketbra{0}{0}^{F_1} \otimes \Gamma^{A_1\to B_1} + \kappa \ketbra{1}{1}^{F_1} \otimes \mc E_1^{A_1\to B_1}.
    \end{equation}
    In particular, we can choose 
    \begin{equation}
    \lambda = (1+\kappa^n)^{-1/n}.
    \end{equation}
\end{prop}
\begin{proof}
    First, denote $$\mc N := \Gamma_\kappa^{A_1 \to F_1B_1} \oplus \mc E_{\lambda}^{A_2 \to F_2B_2},$$
    we apply Lemma \ref{lemma:basic} for direct sum channels: 
    \begin{align*}
        \mc Q^{(1)}(\mc N^{\otimes n}) = \max_{0\le \ell \le n} \mc Q^{(1)}\left((\Gamma_\kappa^{A_1 \to F_1B_1})^{\otimes \ell} \otimes (\mc E_{\lambda}^{A_2 \to F_2B_2})^{\otimes (n-\ell)}\right).
    \end{align*}
    Choosing $\ell \in [1,n-1]$ and $\rho^A := \rho^{A_1^\ell A_2^{n-\ell}}$ as the maximizer, then applying Lemma \ref{lemma:basic} for flagged channels:
    \begin{align}
        \mc Q^{(1)}(\mc N^{\otimes n}) & = I_c\left((\Gamma_\kappa^{A_1 \to F_1B_1})^{\otimes \ell} \otimes (\mc E_{\lambda}^{A_2 \to F_2B_2})^{\otimes (n-\ell)}, \rho^A\right) \\
        & = \sum_{\textbf{b} \in \{0,1\}^n} (1-\kappa)^{\sum_{i=1}^\ell b_i} \kappa^{\ell - \sum_{i=1}^\ell b_i} (1-\lambda)^{\sum_{j=\ell + 1}^n b_i} \lambda^{n-\ell - \sum_{j=\ell + 1}^n b_i}I_c(\bigotimes_{i=1}^\ell \mc M_{b_i} \otimes \bigotimes_{j=\ell+1}^{n-\ell} \mc N_{b_j}, \rho^A), \label{sum:e1}
    \end{align}
    where we denote $$\mc M_0 = \mc E_1^{A_1\to B_1},\ \mc M_1 = \Gamma^{A_1 \to B_1},\ \mc N_0 = \mc E_1^{A_2\to B_2},\ \mc N_1 = \mc I^{A_2 \to B_2} $$
    and we used the calculation for tensor products of flagged channels: 
    \begin{align*}
        (\Gamma_\kappa^{A_1 \to F_1B_1})^{\otimes \ell} = \sum_{\textbf{b} \in \{0,1\}^\ell} (1-\kappa)^{\sum_{i=1}^\ell b_i} \kappa^{\ell - \sum_{i=1}^\ell b_i} \ketbra{\textbf{b}}{\textbf{b}}^{F_1^\ell} \otimes \bigotimes_{i=1}^\ell \mc M_{b_i}
    \end{align*}
    and similar for $(\mc E_{\lambda}^{A_2 \to F_2B_2})^{\otimes (n-\ell)}$. 

    Given $\textbf{b} \in \{0,1\}^n$, we have three cases: (i) $\sum_{i=1}^n b_i = 0$, (ii) $\sum_{i=1}^\ell b_i>0, \sum_{j=\ell + 1}^{n} b_j = 0$ and (iii) $\sum_{j=\ell + 1}^{n} b_j > 0$. Then the calculation of \eqref{sum:e1} can be divided into three cases: 
    \begin{align*}
       & \text{(i)} \sum_{i=1}^n b_i = 0:\quad \kappa^\ell \lambda^{n-\ell} I_c\left((\mc E_1^{A_1\to B_1})^{\otimes \ell} \otimes (\mc E_1^{A_2\to B_2})^{\otimes (n-\ell)}, \rho^A\right) = -\kappa^\ell \lambda^{n-\ell} S(\rho^A). \\
       & \text{(ii)} \sum_{i=1}^\ell b_i>0, \sum_{j=\ell + 1}^{n} b_j = 0: \sum_{\textbf{b}\in \{0,1\}^\ell,|\textbf{b}|\neq 0} (1-\kappa)^{\sum_{i=1}^\ell b_i} \kappa^{\ell - \sum_{i=1}^\ell b_i} \lambda^{n-\ell}I_c\left(\bigotimes_{i=1}^\ell \mc M_{b_i} \otimes (\mc E_1^{A_2\to B_2})^{\otimes (n-\ell)}, \rho^A\right) \\
       & \hspace{4.1cm} \le \sum_{\textbf{b}\in \{0,1\}^\ell,|\textbf{b}|\neq 0} (1-\kappa)^{\sum_{i=1}^\ell b_i} \kappa^{\ell - \sum_{i=1}^\ell b_i} \lambda^{n-\ell}I_c\left((\Gamma^{A_1\to B_1})^{\otimes \ell} \otimes (\mc E_1^{A_2\to B_2})^{\otimes (n-\ell)}, \rho^A\right) \\
       & \hspace{4.1cm} = (1-\kappa^\ell)\lambda^{n-\ell} I_c\left((\Gamma^{A_1\to B_1})^{\otimes \ell} \otimes (\mc E_1^{A_2\to B_2})^{\otimes (n-\ell)}, \rho^A\right) \\
       & \hspace{4.1cm} \le (1-\kappa^\ell)\lambda^{n-\ell} I_c\left((\Gamma^{A_1\to B_1})^{\otimes \ell}, \rho_{A_1^\ell}\right) \le 0. \\ 
       & \text{(iii)} \sum_{j=\ell + 1}^{n} b_j > 0: \sum_{\substack{\exists \ell + 1\le j \le n\\ b_j = 1}} (1-\kappa)^{\sum_{i=1}^\ell b_i} \kappa^{\ell - \sum_{i=1}^\ell b_i} (1-p)^{\sum_{j=\ell + 1}^n b_i} p^{n-\ell - \sum_{j=\ell + 1}^n b_i} I_c\left(\bigotimes_{i=1}^\ell \mc M_{b_i} \otimes \bigotimes_{j=\ell+1}^{n-\ell} \mc N_{b_j}, \rho^A\right) \\
        & \hspace{2cm} \le \sum_{\substack{\exists \ell + 1\le j \le n\\ b_j = 1}} (1-\kappa)^{\sum_{i=1}^\ell b_i} \kappa^{\ell - \sum_{i=1}^\ell b_i} (1-p)^{\sum_{j=\ell + 1}^n b_i} \lambda^{n-\ell - \sum_{j=\ell + 1}^n b_i} I_c\left((\mc I^{A_1\to B_1})^{\otimes \ell} \otimes (\mc I^{A_2\to B_2})^{\otimes (n-\ell)}, \rho^A\right) \\
        &  \hspace{2cm} = (1-\lambda^{n-\ell}) S(\rho^A).
    \end{align*}
    Combining the above three cases, we have 
    \begin{align*}
        \mc Q^{(1)}(\mc N^{\otimes n}) \le (1-\lambda^{n-\ell} - \kappa^\ell \lambda^{n-\ell} )S(\rho^A).
    \end{align*}
    The above is less than zero if $\lambda \ge (1 + \kappa^\ell)^{-\frac{1}{n-\ell}}\in (0,1)$. In particular, we can choose $\lambda = (1+\kappa^n)^{-1/n}$.
\end{proof}
In summary, we fix $n\ge 1$ and $\varepsilon>0$ sufficiently small, we first choose $\kappa \in (0,\frac{1}{2})$ such that $1 - 2h(\varepsilon) - 4\varepsilon - \frac{\kappa}{1-\kappa}>0$. Via Proposition \ref{main:additivity}, we choose $\lambda(n,\kappa)<1$ such that 
\begin{equation}
    \mc Q^{(1)}\left((\Gamma_\kappa^{A' \to F B} \oplus \mc E_{\lambda}^{A_0 \to FA_0})^{\otimes n}\right) = 0.
\end{equation}
Via Theorem \ref{main:non-additivity}, we choose $N > \frac{\log(1 - 2h(\varepsilon) - 4\varepsilon - \frac{\kappa}{1-\kappa}) - \log(h(\frac{1+|c|}{2}))}{\log \lambda}$, we have 
    \begin{align}
        \mc Q^{(1)}\left((\Gamma_\kappa^{A' \to F B} \oplus \mc E_{\lambda}^{A_0 \to FA_0})^{\otimes N}\right) > 0. 
    \end{align}   
Our construction simplifies the erasure channel by omitting the $a_0$ component. By presenting the discussion in a more accessible way, we hope it will stimulate further exploration into the computability of quantum capacity.

\section{Conclusion and Outlook}
This work develops a unified picture of quantum capacity amplification built from private-state induced channels. In Section~\ref{sec:amp} we proved quantitative lower bounds on the achievable coherent information under explicit constructions (see, e.g., Theorem~\ref{thm:amplification}), and exhibited both amplification and super-amplification regimes, including concrete illustrations for erasure and depolarizing channels. Assuming the Spin Alignment Conjecture (SAC)~\cite{leditzky2023platypus}, Section~\ref{sec:single letter} established a single-letter formula for the quantum capacity of a broad class of private channels that are neither degradable, anti-degradable, nor PPT, and leveraged this to build channels with vanishing quantum capacity yet unbounded private capacity. Finally, Section~\ref{sec:approximate} extended the picture to approximate private channels, proving metric separation from the anti-degradable set (via diamond-norm lower bounds and continuity), giving an alternative superactivation proof in the approximate setting (extending Smith--Yard~\cite{Smith_2008}), and clarifying the scope of computability results in light of~\cite{Cubitt_2015}: these constructions do not settle the computability of the regularization of quantum capacity.

\emph{Limitations and outlook for future work}:
\begin{itemize}
    \item Several statements rely on SAC, with partial solution given in~\cite{Alhejji_2025, Alhejji_2024}. A complete solution is still open. 
    \item Amplification constants in our quantitative bounds are likely suboptimal, since it is given by maximally entangled state ansatz.
    \item Most constructions are built around a special class of private states, leaving open how universal the phenomenon is beyond this template. A general open question: suppose $\mc P(\mc N)> \mc Q(\mc N)$, does there exist an anti-degradable (or zero-capacity) channel $\mc M$ such that $\mc Q(\mc N\otimes \mc M) > \mc Q(\mc N)$?
\end{itemize}
\indent We hope these results and questions stimulate a more systematic theory of how private structure governs coherent transmission. 

\section*{Acknowledgement}
Peixue Wu thanks Graeme Smith, Debbie Leung, Vikesh Siddhu, Zhen Wu for helpful discussions. 

\bibliographystyle{marcotomPB} 
\bibliography{privacy}
\end{document}